\documentclass[acmnow,pdftex]{acmtrans2m}
\pdfoutput=1

\let\orgsetcounter\setcounter
\usepackage{calc}
\usepackage{xspace,amstext,amssymb,url}
\usepackage{amsmath}
\usepackage{algorithm}
\usepackage[noend]{algorithmic}
\usepackage{pgf}
\usepackage{tikz}
\usepackage[hang]{subfigure}

\newcommand{\exre}{((b?(a+c))^+d)^+e}
\newcommand{\soa}{{SOA}\xspace}
\newcommand{\soas}{{SOAs}\xspace}

\DeclareMathOperator{\occ}{\mathit{occ}}

\newcommand{\mrk}[2]{\ensuremath{{{#1}^{(#2)}}}}
\DeclareMathOperator{\strip}{\mathit{strip}}
\DeclareMathOperator{\pos}{\mathit{pos}}
\DeclareMathOperator{\first}{\mathit{first}}
\DeclareMathOperator{\last}{\mathit{last}}
\DeclareMathOperator{\follow}{\mathit{follow}}

\newcommand{\cure}{\textnormal{\sc rwr}\xspace}
\newcommand{\kcure}{\ensuremath{\textnormal{{\sc rwr}}^2}\xspace}

\DeclareMathOperator{\rep}{\mathit{rep}}

\newcommand{\res}{\text{regular expressions}\xspace}

\newcommand{\oa}{\textnormal{-\textsc{OA}}\xspace}
\newcommand{\koa}{\ensuremath{\mathnormal{k}}\text{-\textsc{OA}}\xspace}
\newcommand{\koas}{\ensuremath{\mathnormal{k}}\textnormal{-\textsc{OA}s}\xspace}
\newcommand{\oas}{\text{-\textsc{OA}s}\xspace}
\newcommand{\ore}{\text{-\textsc{ORE}}\xspace}
\newcommand{\ores}{\text{-\textsc{ORE}s}\xspace}
\newcommand{\kore}{\ensuremath{\mathnormal{k}\text{-\textsc{ORE}}}\xspace}

\newcommand{\kores}{\ensuremath{\mathnormal{k}\text{-\textsc{ORE}s}}\xspace}
\newcommand{\sore}{\text{SORE}\xspace}
\newcommand{\sores}{\text{SOREs}\xspace}

\newcommand{\dtds}{\text{DTDs}\xspace}
\newcommand{\xsd}{\text{XSD}\xspace}
\newcommand{\xsds}{\text{XSDs}\xspace}

\newcommand{\pomm}{\text{POMM}\xspace}

\newcommand{\alphabet}{\ensuremath{\Sigma}}
\newcommand{\emptystr}{\ensuremath{\varepsilon}}
\DeclareMathOperator{\lang}{\mathcal{L}}
\DeclareMathOperator{\con}{.}

\newcommand{\elem}[1]{\ensuremath{\mathtt{#1}}}

\newcommand{\car}[1]{\ensuremath{|#1|}}

\DeclareMathOperator{\lab}{\mathit{lab}}
\newcommand{\prob}[1]{\ensuremath{P[#1]}}

\DeclareMathOperator{\src}{\mathit{src}}
\DeclareMathOperator{\sink}{\mathit{sink}}

\newcommand{\vldb}{\ensuremath{\textnormal{{\sc rwr}}^0}\xspace}

\newcommand{\Corpus}{S}
\newcommand{\Disambiguate}{\textnormal{\sc Disambiguate}\xspace}
\newcommand{\Prune}{\textnormal{\sc Prune}\xspace}
\newcommand{\BaumWelsh}{\textnormal{\sc BaumWelsh}\xspace}
\newcommand{\KoaToKore}{\kcure\xspace}

\newcommand{\ixsd}{\ensuremath{i}\textnormal{\sc XSD}\xspace}

\newcommand{\iKoa}{\ensuremath{i}\textnormal{\sc Koa}\xspace}
\newcommand{\learn}{\textnormal{\ensuremath{i}{\sc DRegEx}}\xspace}
\newcommand{\fixed}{\ensuremath{i\textnormal{{\sc DRegEx}}^{\mathrm{fixed}}}\xspace}

\DeclareMathOperator{\Out}{\Succ}

\DeclareMathOperator{\Pred}{Pred}
\DeclareMathOperator{\Succ}{Succ}

\DeclareMathOperator{\First}{first}

\DeclareMathOperator{\Init}{init}
\DeclareMathOperator{\Best}{best}
\DeclareMathOperator{\Data}{datacost}
\DeclareMathOperator{\Size}{Density}

\newcommand{\myparagraph}[1]{\vspace{2ex} \noindent \textbf{#1.}}

\newcommand{\onlyInOnlineVersion}[1]{}

\newtheorem{theorem}{Theorem}[section]

\newtheorem{proposition}[theorem]{Proposition} 
 
\newdef{definition}[theorem]{Definition} 

\title{Learning Deterministic Regular Expressions for the Inference of
  Schemas from XML Data}

\author{GEERT JAN BEX, WOUTER GELADE, FRANK NEVEN\\Hasselt
  University and Transnational University of Limburg 
\and 
STIJN VANSUMMEREN \\
  Universit\'e Libre de Bruxelles}

\begin{abstract} 
  Inferring an appropriate DTD or XML Schema Definition (XSD) for a
  given collection of XML documents essentially reduces to learning
  \emph{deterministic} regular expressions from sets of positive
  example words.  Unfortunately, there is no algorithm capable of
  learning the complete class of deterministic regular expressions
  from positive examples only, as we will show. The regular
  expressions occurring in practical DTDs and XSDs, however, are such
  that every alphabet symbol occurs only a small number of times. As
  such, in practice it suffices to learn the subclass of deterministic
  regular expressions in which each alphabet symbol occurs at most $k$
  times, for some small $k$. We refer to such expressions as
  $k$-occurrence regular expressions ($k\ores$ for short).  Motivated
  by this observation, we provide a probabilistic algorithm that
  learns $k\ores$ for increasing values of $k$, and selects the
  deterministic one that best describes the sample based on a Minimum
  Description Length argument.  The effectiveness of the method is
  empirically validated both on real world and synthetic
  data. Furthermore, the method is shown to be conservative over the
  simpler classes of expressions considered in previous work.
\end{abstract}




\begin{document}

{\let\setcounter\orgsetcounter
\begin{bottomstuff}
  A preliminary version of this article appeared in the 17th
  International World Wide Web Conference (WWW 2008).
\end{bottomstuff}
}
\maketitle

\section{Introduction}
\label{sec:introduction}

\newcommand{\xtract}{{\sc xtract}\xspace}
\newcommand{\XTRACT}{{\sc xtract}\xspace}

Recent studies stipulate that schemas accompanying collections of XML
documents are sparse and erroneous in practice.  Indeed,
\citeN{Barb05} and \citeN{Mign03} have shown that approximately half
of the XML documents available on the web do not refer to a schema.
In addition, \citeN{Bex04} and \citeN{Mart06b} have noted that about
two-thirds of XML Schema Definitions (XSDs) gathered from schema
repositories and from the web at large are not valid with respect to
the W3C XML Schema specification~\cite{XSDS01}, rendering them
essentially useless for immedidate application.  A similar observation
was made by \citeN{Sahu00} concerning Document Type Definitions
(DTDs).  Nevertheless, the presence of a schema strongly facilitates
optimization of XML processing (cf., e.g.,
\cite{geerts2005,tamer,shrex,Frei02,Koch04b,kossmann,nevenschwentickicdt03})
and various software development tools such as Castor~\cite{castor}
and SUN's JAXB~\cite{jaxb} rely on schemas as well to perform
object-relational mappings for persistence.  Additionally, the
existence of schemas is imperative when integrating (meta) data
through schema matching~\cite{schemamatching} and in the area of
generic model management~\cite{mmbern}.

Based on the above described benefits of schemas and their
unavailability in practice, it is essential to devise algorithms that
can infer a DTD or XSD for a given collection of XML documents when
none, or no syntactically correct one, is present. This is also
acknowledged by \citeN{Flor05} who emphasizes that in the context of
data integration
\begin{quote}
  \it ``We need to extract good-quality schemas
  automatically from existing data and perform incremental maintenance
  of the generated schemas.''
\end{quote}

As illustrated in Figure~\ref{fig:dtd-store}, a DTD is essentially a
mapping $d$ from element names to regular expressions over element
names. An XML document is valid with respect to the DTD if for every
occurrence of an element name $e$ in the document, the word formed by
its children belongs to the language of the corresponding regular
expression $d(e)$. For instance, the DTD in Figure~\ref{fig:dtd-store}
requires each $\mathtt{store}$ element to have zero or more
$\mathtt{order}$ children, which must be followed by a
$\mathtt{stock}$ element. Likewise, each order must have a
$\mathtt{customer}$ child, which must be followed by one or more
$\mathtt{item}$ elements.

\begin{figure}[t]
  \centering
\[ 
\begin{array}{l}
\texttt{<!}\mathtt{ELEMENT}\ \mathtt{store}\  (\mathtt{order^*,
  stock})\texttt{>}\\
\texttt{<!}\mathtt{ELEMENT}\ \mathtt{order}\  (\mathtt{customer,
  item^+})\texttt{>}\\
\texttt{<!}\mathtt{ELEMENT}\ \mathtt{customer}\  (\mathtt{first, last, email^*})\texttt{>}\\
\texttt{<!}\mathtt{ELEMENT}\ \mathtt{item}\  (\mathtt{id, price +
  (qty, (supplier + item^+))})\texttt{>}\\
\texttt{<!}\mathtt{ELEMENT}\ \mathtt{stock}\  (\mathtt{item^*})\texttt{>}\\
\texttt{<!}\mathtt{ELEMENT}\ \mathtt{supplier}\  (\mathtt{first, last, email^*})\texttt{>}
\end{array}
\]
  \caption{An example DTD.}
  \label{fig:dtd-store}
\end{figure}

To infer a DTD from a corpus of XML documents $\mathcal{C}$ it hence
suffices to look, for each element name $e$ that occurs in a document
in $\mathcal{C}$, at the set of element name words that occur below
$e$ in $\mathcal{C}$, and to infer from this set the corresponding
regular expression $d(e)$. As such, the inference of DTDs reduces to
the inference of regular expressions from sets of positive example
words. To illustrate, from the words $\mathtt{id\ price}$, $\mathtt{id\
  qty\ supplier}$, and $\mathtt{id\ qty\ item\ item}$ appearing under
$\texttt{<}\mathtt{item}\texttt{>}$ elements in a sample XML corpus,
we could derive the rule
\[ \mathtt{item} \to  (\mathtt{id, price +
  (qty, (supplier + item^+))}).
\]
Although XSDs are more expressive than DTDs, and although XSD
inference is therefore more involved than DTD inference, derivation of
regular expressions remains one of the main building blocks on which
XSD inference algorithms are built. In fact, apart from also inferring
atomic data types, systems like Trang~\cite{trang} and
XStruct~\cite{Hege06} simply infer DTDs in XSD syntax. The more recent
$\ixsd$ algorithm~\cite{Bex07} does infer true \xsd schemas by first
deriving a regular expression for every {\em context} in which an
element name appears, where the context is determined by the path from
the root to that element, and subsequently reduces the number of
contexts by merging similar ones.

So, the effectiveness of DTD or XSD schema inference algorithms is
strongly determined by the accuracy of the employed regular expression
inference method. The present article presents a method to reliably
learn regular expressions that are far more complex than the classes
of expressions previously considered in the literature.

\subsection{Problem setting}
\label{sec:problem-setting}

In particular, let $\alphabet$ be a fixed set of alphabet symbols
(also called element names), and let $\alphabet^*$ be the set of all
words over $\alphabet$.

\begin{definition}[(Regular Expressions)]
  Regular expressions are derived by the following grammar.
  \[ r,s ::= \emptyset \mid \emptystr \mid a \mid r\con s \mid r+s
  \mid r? \mid r^+ \] Here, parentheses may be added to avoid
  ambiguity; $\emptystr$ denotes the empty word; $a$ ranges over
  symbols in $\alphabet$; $r\con s$ denotes concatenation; $r + s$
  denotes disjunction; $r^+$ denotes one-or-more repetitions; and $r?$
  denotes the optional regular expression. That is, the language
  $\lang(r)$ accepted by regular expression $r$ is given by:
  \begin{align*}
    \lang(\emptyset) & = \emptyset & 
    \lang(\emptystr) & = \{\emptystr\} \\
    \lang(a) & = \{a\} &
    \lang(r \con s) & = \{ vw \mid v \in \lang(r), w
  \in \lang(s) \} \\ 
  \lang(r + s) & = \lang(r) \cup \lang(s) & 
  \lang(r^+) & = \{ v_1
  \dots v_n \mid n \geq 1 \text{ and } v_1,\dots,v_n \in \lang(r) \}\\
  \lang(r?) & = \lang(r) \cup \{\emptystr\}.
  \end{align*}
\end{definition}

  \vspace{-0.6cm}\hspace{11.7cm} \qed
\medskip

Note that the Kleene star operator (denoting zero or more repititions
as in $r^*$) is not allowed by the above syntax. This is not a
restriction, since $r^*$ can always be represented as $(r^+)?$ or
$(r?)^+$. Conversely, the latter can always be rewritten into the
former for presentation to the user. 

The class of \emph{all} regular expressions is actually too large for
our purposes, as both \dtds and \xsds require the regular expressions
occurring in them to be \emph{deterministic} (also sometimes called
one-unambiguous~\cite{Brue98}).  Intuitively, a regular expression is
deterministic if, without looking ahead in the input word, it allows
to match each symbol of that word uniquely against a position in the
expression when processing the input in one pass from left to right.
For instance, $(a+b)^*a$ is not deterministic as already the first
symbol in the word $aaa$ could be matched by either the first or the
second $a$ in the expression. Without lookahead, it is impossible to
know which one to choose. The equivalent expression $b^*a(b^*a)^*$, on
the other hand, is deterministic. 

\begin{definition}
  \label{def:deterministic}
  Formally, let $\overline{r}$ stand for the regular expression
  obtained from $r$ by replacing the $i$th occurrence of alphabet
  symbol $a$ in $r$ by $\mrk{a}{i}$, for every $i$ and $a$.  For
  example, for $r = b^+ a (b a^+)?$ we have $\overline{r} =
  \mrk{b}{1}^+ \mrk{a}{1}  (\mrk{b}{2} \mrk{a}{2}^+)?$.  A regular
  expression $r$ is \emph{deterministic} if there are no words
  $w\mrk{a}{i}v$ and $w\mrk{a}{j}v'$ in $\lang(\overline{r})$ such
  that~$i \neq j$. \hspace{10.2cm} \qed
\end{definition}
Equivalently, an expression is deterministic if the Glushkov
construction~\cite{Brug93} translates it into a deterministic finite
automaton rather than a non-deterministic one~\cite{Brue98}.  Not
every non-deterministic regular expression is equivalent to a
deterministic one~\cite{Brue98}. Thus, semantically, the class of
deterministic regular expressions forms a strict subclass of the class
of all regular expressions.  

For the purpose of inferring DTDs and XSDs from XML data, we are hence
in search of an algorithm that, given enough sample words of a target
deterministic regular expression $r$, returns a deterministic
expression $r'$ equivalent to $r$. In the framework of \emph{learning
  in the limit}~\cite{Gold67}, such an algorithm is said to learn the
deterministic regular expressions from positive data.

\begin{definition}
  Define a \emph{sample} to be a finite subset of $\alphabet^*$ and
  let $\mathcal{R}$ be a subclass of the regular expressions. An
  algorithm $M$ mapping samples to expressions in $\mathcal{R}$ 
  \emph{learns $\mathcal{R}$ in the limit from positive data}
  if (1) $S \subseteq \lang(M(S))$ for every sample $S$ and (2) to
  every $r \in \mathcal{R}$ we can associate a so-called
  \emph{characteristic sample} $S_r \subseteq \lang(r)$ such that, for
  each sample $S$ with $S_r \subseteq S \subseteq \lang(r)$, $M(S)$ is
  equivalent to $r$. \hspace{1.2cm} \qed
\end{definition}

Intuitively, the first condition says that $M$ must be \emph{sound};
the second that $M$ must be \emph{complete}, given enough data. A
class of regular expressions $\mathcal{R}$ is \emph{learnable in the
  limit from positive data} if an algorithm exists that learns
$\mathcal{R}$.  For the class of all regular expressions, it was shown
by Gold that no such algorithm exists~\cite{Gold67}. We extend this
result to the class of deterministic expressions:
\begin{theorem}
  \label{THM:DREG-NOTLEARNABLE}
  The class of deterministic regular expressions is not learnable in
  the limit from positive data.
\end{theorem}

\begin{proof}
  It was shown by \citeN[Theorem I.8]{Gold67}, that any class of
  regular expressions that contains all non-empty finite languages as
  well as at least one infinite language is not learnable in the limit
  from positive data.  Since deterministic regular expressions like
  $a^*$ define an infinite language, it suffices to show that every
  non-empty finite language is definable by a deterministic
  expression. Hereto, let $S$ be a finite, non-empty set of words. Now
  consider the prefix tree $T$ for $S$. For example, if $S = \{a, aab,
  abc, aac\}$, we have the following prefix tree:
\begin{center}
\includegraphics[viewport=136 604 204 675]{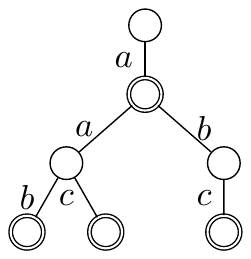}
\end{center}
Nodes for which the path from the root to that node forms a word in
$S$ are marked by double circles. In particular, all leaf nodes are
marked.

By viewing the internal nodes in $T$ with two or more children as
disjunctions; internal nodes in $T$ with one child as conjunctions;
and adding a question mark for every marked internal node in $T$, it
is straightforward to transform $T$ into a regular expression. For
example, with $S$ and $T$ as above we get $r = a \con(b\con c + a
\con(b +c))?$. Clearly, $\lang(r) = S$. Moreover, since no node in $T$
has two edges with the same label, $r$ must be deterministic.
\end{proof}

Theorem~\ref{THM:DREG-NOTLEARNABLE} immediately excludes the
possibility for an algorithm to infer the full class of \dtds or
\xsds. In practice, however, regular expressions occurring in \dtds
and \xsds are \emph{concise} rather than arbitrarily complex. Indeed,
a study of $819$ DTDs and XSDs gathered from the Cover
Pages~\cite{Cover03} (including many high-quality XML standards) as
well as from the web at large, reveals that \res occurring in
practical schemas are such that every alphabet symbol occurs only a
small number of times~\cite{Mart06b}. In practice, therefore, it
suffices to learn the subclass of deterministic regular expressions in
which each alphabet symbol occurs at most $k$ times, for some small
$k$. We refer to such expressions as \emph{$k$-occurrence regular
  expressions}.
\begin{definition}
  A regular expression is \emph{$k$-occurrence} if every alphabet
  symbol occurs at most $k$ times in it. \hspace{7.5cm} \qed
\end{definition}
For example, the expressions $\elem{customer}\con \elem{order}^+$ and
$(\elem{school} + \elem{institute})^+$ are both $1$-occurrence, while
$\elem{id}\con(\elem{qty} + \elem{id})$ is $2$-occurrence (as
$\elem{id}$ occurs twice). Observe that if $r$ is $k$-occurrence, then
it is also $l$-occurrence for every $l \geq k$. To simplify notation
in what follows, we abbreviate `$k$-occurrence regular expression' by
\kore and also refer to the $1$\ores as `single occurrence regular
expressions' or \sores.

\subsection{Outline and Contributions}
\label{sec:outl-contr}

Actually, the above mentioned examination shows that in the majority
of the cases $k=1$. Motivated by that observation, we have studied and
suggested practical learning algorithms for the class of deterministic
\sores in a companion article \cite{Bex06}. These algorithms, however,
can only output \sores even when the target regular expression is
not. In that case they always return an approximation of the target
expressions. It is therefore desirable to also have learning
algorithms for the class of deterministic $\kores$ with $k \geq 2$.
Furthermore, since the exact $k$-value for the target expression,
although small, is unknown in a schema inference setting, we also
require an algorithm capable of determining the best value of $k$
automatically.

We begin our study of this problem in Section~\ref{sec:basic-results} by
showing that, for each fixed $k$, the class of deterministic $\kores$
\emph{is} learnable in the limit from positive examples only.  We also
argue, however, that this theoretical algorithm is unlikely to work
well in practice as it does not provide a method to automatically
determine the best value of $k$ and needs samples whose size can be
exponential in the size of the alphabet to successfully learn some
target expressions.

In view of these observations, we provide in
Section~\ref{sec:mach-learn-appr} the practical algorithm $\learn$.
Given a sample of words $S$, $\learn$ derives corresponding
deterministic \kores for increasing values of $k$ and selects from
these candidate expressions the expression that describes $S$ best. To
determine the ``best'' expression we propose two measures: (1) a
Language Size measure and (2) a Minimum Description Length measure
based on the work of Adriaans and Vit\'anyi~\citeyear{Adri06a}.  The
main technical contribution lies in the subroutine used to derive the
actual $\kores$ for $S$. Indeed, while for the special case where
$k=1$ one can derive a $\kore$ by first learning an automaton $A$ for
$S$ using the inference algorithm of \citeN{Garc90}, and by
subsequently translating $A$ into a $1\ore$ (as shown in
\cite{Bex06}), this approach does not work when $k \geq 2$. In
particular, the algorithm of Garcia and Vidal only works when learning
languages that are ``$n$-testable'' for some fixed natural number
$n$~\cite{Garc90}.  Although every language definable by a $1\ore$ is
$2$-testable~\cite{Bex06}, there are languages definable by a $2\ore$,
for instance $a^*ba^*$, that are not $n$-testable for any $n$. We
therefore use a probabilistic method based on Hidden Markov Models to
learn an automaton for $S$, which is subsequently translated into a
$\kore$.

The effectiveness of $\learn$ is empirically validated in
Section~\ref{sec:experiments} both on real world and synthetic data.
We compare the results of $\learn$ with those of the algorithm
presented in previous work~\cite{Bex08}, to which we refer as
$\learn(\vldb)$.

\section{Related Work}
\label{sec:related-work}

\vspace{-1ex}
\myparagraph{Semi-structured data} In the context of semi-structured
data, the inference of schemas as defined in~\cite{unstructured,lore}
has been extensively studied~\cite{dataguides,nestorov}. No methods
were provided to translate the inferred types to regular expressions,
however.

\myparagraph{DTD and XSD inference} In the context of DTD inference,
\citeN{Bex06} gave in earlier work two inference algorithms: one for
learning $1\ores$ and one for learning the subclass of $1\ores$ known
as \emph{chain regular expressions}.  The latter class can also be
learned using Trang~\cite{trang}, state of the art software written by
James Clark that is primarily intended as a translator between the
schema languages DTD, Relax NG~\cite{RELAXNG01}, and XSD, but also
infers a schema for a set of XML documents. In contrast, our goal in
this article is to infer the more general class of deterministic
expressions.  \xtract~\cite{Garo03} is another regular expression
learning system with similar goals. We note that \xtract also uses the
Minimum Description Length principle to choose the best expression
from a set of candidates.

Other relevant DTD inference research is \cite{wongDTD} and
\cite{DBLP:conf/krdb/Chidlovskii01} that learn finite automata but do
not consider the translation to deterministic regular expressions.
Also, in \cite{Youn00} a method is proposed to infer DTDs through
stochastic grammars where right-hand sides of rules are represented by
probabilistic automata.  No method is provided to transform these into
regular expressions. Although \citeN{ahonen} proposes such a
translation, the effectiveness of her algorithm is only illustrated by
a single case study of a dictionary example; no experimental study is
provided.

Also relevant are the XSD inference systems~\cite{Bex07,trang,Hege06}
that, as already mentioned, rely on the same methods for learning
regular expressions as DTD inference.

\myparagraph{\bf Regular expression inference} Most of the learning of
regular languages from positive examples in the computational learning
community is directed towards inference of automata as opposed to
inference of \res~\cite{angluinsmith,Pitt89,Saka97}. However,
these approaches learn strict subclasses of the regular languages
which are incomparable to the subclasses considered here. Some
approaches to inference of regular expressions for restricted cases
have been considered. For instance, \cite{Braz93} showed
that regular expressions without union can be approximately learned in
polynomial time from a set of examples satisfying some
criteria. \cite{fernaualt} provided a learning algorithm for
regular expressions that are finite unions of pairwise left-aligned
union-free regular expressions.  The development is purely
theoretical, no experimental validation has been performed.

\myparagraph{\bf HMM learning} Although there has been work on Hidden
Markov Model structure induction~\cite{Rabi89,Frei00}, the requirement
in our setting that the resulting automaton is deterministic is, to the
best of our knowledge, unique.



\section{Basic results}
\label{sec:basic-results}

In this section we establish that, in contrast to the class of all
deterministic expressions, the subclass of deterministic \kores
\emph{can} theoretically be learned in the limit from positive
data, for each fixed $k$. We also argue, however, that this
theoretical algorithm is unlikely to work well in practice.

Let $\alphabet(r)$ denote the set of alphabet symbols that occur in a
regular expression $r$, and let $\alphabet(S)$ be similarly defined
for a sample $S$. Define the \emph{length} of a regular expression $r$
as the length of it string representation, including operators and
parenthesis. For example, the length of $(a \con b)^+? + c$ is $9$.

\begin{theorem}
  \label{THM:KORE-LEARNABLE}
  For every $k$ there exists an algorithm $M$ that learns the class of
  deterministic $\kores$ from positive data. Furthermore, on input
  $S$, $M$ runs in time polynomial in the size of $S$, yet exponential
  in $k$ and $|\alphabet(S)|$.
\end{theorem}

\begin{proof}
  The algorithm $M$ is based on the following observations. First
  observe that every deterministic \kore $r$ over a finite alphabet
  $A\subseteq \alphabet$ can be simplified into an equivalent
  deterministic \kore $r'$ of length at most $10k|A|$ by rewriting $r$
  according to the following system of rewrite rules until no more
  rule is applicable:
    \[  \begin{array}{rcl@{\qquad}rcl}
      ((s)) & \to & (s) & s?^+ & \to & s^+? \\
    s?? & \to & s? & 
    s^{++} & \to & s^+ \\
    s + \emptystr
    & \to & s? & 
    \emptystr + s & \to & s? \\ 
    s \con \emptystr & \to & s &
    \emptystr \con s & \to & s \\
    \emptystr? & \to & \emptystr  &
    \emptystr^+ & \to & \emptystr \\
    s + \emptyset & \to & s &
    \emptyset + s & \to  & s \\
    s \con \emptyset & \to & \emptyset &
    \emptyset \con s & \to & \emptyset \\
    \emptyset? & \to & \emptyset &
    \emptyset^+ &\to & \emptyset
  \end{array}\]
  (The first rewrite rule removes redundant parenthesis in $r$.)
  Indeed, since each rewrite rule clearly preserves determinism and language
  equivalence, $r'$ must be a deterministic expression equivalent to
  $r$. Moreover, since none of the rewrite rules
  duplicates a subexpression and since $r$ is a \kore, so is $r'$. Now
  note that, since no rewrite rule applies to it, $r'$ is either $\emptyset$, $\emptystr$, or generated by
  the following grammar
  \[
  \begin{array}{rcl}
    t & ::= & a \mid a? \mid a^+ \mid a^+? \mid (a) \mid (a)? \mid
    (a)^+ \mid (a)^+? \\ & \mid & t_1 \con t_2 \mid (t_1
    \con t_2) \mid (t_1 \con t_2)? \mid (t_1 \con t_2)^+ \mid
    (t_1 \con t_2)^+? \\
    & \mid & t_1 + t_2 \mid (t_1 + t_2) \mid (t_1 + t_2)? \mid (t_1 + t_2)^+
    \mid (t_1 + t_2)^+? 
  \end{array}\]    
  It  is not difficult to verify by structural induction that any expression $t$
  produced by this grammar has length \[ |t| \leq - 4 + 10 \sum_{a \in
    \alphabet(t)} \rep(t, a),\] where $\rep(t,a)$ denotes the  number of
  times  alphabet symbol $a$ occurs in $t$. For instance, $\rep(b\con
  (b+c), a) = 0$ and $\rep(b \con (b+c), b) = 2$. Since
  $\rep(r', a) \leq k$ for every $a \in \alphabet(r')$, it readily follows
  that $|r'| \leq
  10k|A| - 4 \leq 10k|A|$. 

  Then observe that all possible regular expressions over $A$ of
  length at most $10k|A|$ can be enumerated in time exponential in
  $k|A|$.  Since checking whether a regular expression is
  deterministic is decidable in polynomial time~\cite{Brue98}; and
  since equivalence of deterministic expressions is decidable in
  polynomial time ~\cite{Brue98}, it follows by the above observations
  that for each $k$ and each finite alphabet $A \subseteq \alphabet$
  it is possible to compute in time exponential in $k|A|$ a finite set
  $\mathcal{R}_A$ of pairwise non-equivalent deterministic \kores over
  $A$ such that
  \begin{itemize}
  \item every $r \in \mathcal{R}_A$ is of size at most $10 k |A|$; and
  \item for every deterministic \kore $r$ over $A$ there exists an
    equivalent expression $r' \in \mathcal{R}_A$.
  \end{itemize}
  (Note that since $\mathcal{R}_A$ is computable in time exponential
  in $k |A|$, it has at most an exponential number of elements in
  $k|A|$.)  Now fix, for each finite $A \subseteq \alphabet$ an
  arbitrary order $\prec$ on $\mathcal{R}_A$, subject to the provision
  that $r \prec s$ only if $\lang(s) - \lang(r) \not = \emptyset$.
  Such an order always exists since $\mathcal{R}_A$ does not contain
  equivalent expressions.

  Then let $M$ be the algorithm that, upon sample $S$, computes
  $\mathcal{R}_{\alphabet(S)}$ and outputs the first (according to $\prec$)
  expression $r \in \mathcal{R}_{\alphabet(S)}$ for which $S \subseteq
  L(r)$.  Since $\mathcal{R}_{\alphabet(S)}$ can be computed in time exponential
  in $k|\alphabet(S)|$; since there are at most an exponential number
  of expressions in $\mathcal{R}_{\alphabet(S)}$; since each
  expression $r \in \mathcal{R}_{\alphabet(S)}$ has size at most $10 k
  |\alphabet(S)|$; and since checking membership in $\lang(r)$ of a
  single word $w \in S$ can be done in time polynomial in the size of
  $w$ and $r$, it follows that $M$ runs in time polynomial in $S$ and
  exponential in $k|\alphabet(S)|$.

  Furthermore, we claim that $M$ learns the class of deterministic
  \kores. Clearly, $S \subseteq \lang(M(S))$ by definition. Hence, it
  remains to show completeness, i.e., that we can associate to each
  deterministic \kore $r$ a sample $S_r \subseteq L(r)$ such that, for
  each sample $S$ with $S_r \subseteq S \subseteq L(r)$, $M(S)$ is
  equivalent to $r$.  Note that, by definition of
  $\mathcal{R}_{\alphabet(r)}$, there exists a deterministic \kore $r'
  \in \mathcal{R}_{\alphabet(r)}$ equivalent to $r$. Initialize
  $S_{r}$ to an arbitrary finite subset of $\lang(r) = \lang(r')$ such
  that each alphabet symbol of $r$ occurs at least once in $S$, i.e.,
  $\alphabet(S_{r}) = \alphabet(r)$. Let $r_1 \prec \dots \prec r_n$
  be all predecessors of $r'$ in $\mathcal{R}_{\alphabet(r)}$
  according to $\prec$.  By definition of $\prec$, there exists a word
  $w_i \in \lang(r) - \lang(r_i)$ for every $1 \leq i \leq n$. Add all
  of these words to $S_r$.  Then clearly, for every sample $S$ with
  $S_r \subseteq S \subseteq \lang(r)$ we have $\alphabet(S) =
  \alphabet(r)$ and $S \not \subseteq L(r_i)$ for every $1 \leq i \leq
  n$.  Since $M(S)$ is the first expression in
  $\mathcal{R}_{\alphabet(r)}$ with $S \subseteq L(r)$, we hence have $M(S)
  = r' \equiv r$, as desired.
\end{proof}

While Theorem~\ref{THM:KORE-LEARNABLE} shows that the class of
deterministic \kores is better suited for learning from positive data
than the complete class of deterministic expressions, it does not
provide a useful practical algorithm, for the following reasons.

\begin{enumerate}
\item First and foremost, $M$ runs in time exponential in the size of
  the alphabet $\alphabet(S)$, which may be problematic for the
  inference of schema's with many element names.
\item Second, while Theorem~\ref{THM:KORE-LEARNABLE} shows that the
  class of deterministic \kores is learnable in the limit for each
  fixed $k$, the schema inference setting is such that we do not know
  $k$ a priori. If we overestimate $k$ then $M(S)$ risks being an
  under-approximation of the target expression $r$, especially when
  $S$ is incomplete. To illustrate, consider the $1\ore$ target
  expression $r = a^+b^+$ and sample $S = \{ab, abbb, aabb\}$.  If we
  overestimate $k$ to, say, $2$ instead of $1$, then $M$ is free to
  output $a a? b^+$ as a sound answer. On the other hand, if we
  underestimate $k$ then $M(S)$ risks being an over-approximation of
  $r$.  Consider, for instance, the $2\ore$ target expression $r =
  aa?b^+$ and the same sample $S = \{ab, abbb, \allowbreak aabb\}$. If
  we underestimate $k$ to be $1$ instead of $2$, then $M$ can only
  output $1\ores$, and needs to output at least $a^+b^+$ in order to
  be sound. In summary: we need a method to determine the most
  suitable value of $k$.
\item Third, the notion of learning in the limit is a very liberal
  one: correct expressions need only be derived when sufficient data
  is provided, i.e., when the input sample is a superset of the
  characteristic sample for the target expression $r$.  The following
  theorem shows that there are reasonably simple expressions $r$ such
  that characteristic sample $S_r$ of any sound and complete learning
  algorithm is at least exponential in the size of $r$. As such, it is
  unlikely for any sound and complete learning algorithm to behave
  well on real-world samples, which are typically incomplete and hence
  unlikely to contain all words of the characteristic sample.
\end{enumerate}

\begin{theorem}\label{THM:KORE-EXP-DATA}
  Let $A = \{a_1,\dots,a_n\} \subseteq \Sigma$ consist of $n$ distinct
  element names. Let $r_1 = (a_1 a_2 + a_3 + \dots + a_n)^+$, and let
  $r_2 = (a_2 + \dots + a_n)^+ a_1 (a_2 + \dots + a_n)^+$. For any
  algorithm that learns the class of deterministic $(2n + 3)\ores$ and
  any sample $S$ that is characteristic for $r_1$ or $r_2$ we have
  $\car{S} \geq \sum_{i=1}^{n}(n-2)^i$.
\end{theorem}
\begin{proof}
  First consider $r_1 = (a_1a_2 + a_3 + \dots + a_n)^+$.  Observe that
  there exist an exponential number of deterministic $(2n +3)\ores$
  that differ from $r_1$ in only a single word. Indeed, let $B = A -
  \{a_1,a_2\}$ and let $W$ consist of all non-empty words $w$ over $B$
  of length at most $n$. Define, for every word $w = b_1\dots b_m \in
  W$ the deterministic $(2n + 3)\ore$ $r_w$ such that $\lang(r_w) =
  \lang(r_1) - \{w\}$ as follows. First, define, for every $1 \leq i
  \leq m$ the deterministic $2\ore$ $r_w^i$ that accepts all words in
  $\lang(r_1)$ that do not start with $b_i$:
  \[ r_w^i := (a_1 a_2 + (B - \{b_i\})) \con (a_1 a_2 + a_3 + \dots + a_n)^* \]
  Clearly, $v \in \lang(r_1) - \{w\}$ if, and only if, $v \in
  \lang(r_1)$ and there is some
  $0 \leq i \leq m$ such that $v$ agrees with $w$ on the first $i$
  letters, but differs in the $(i+1)$-th letter. Hence, it suffices to take
  \[ r_w := r_w^1 + b_1(\emptystr + r_w^2 + b_2(\emptystr + r_w^3 +
  b_3( \dots + b_{m-1}(\emptystr + r_w^{m} + b_m \con r_1) \dots
  ))) \] Now assume that algorithm $M$ learns the class of
  deterministic $(2n+3)\ores$ and suppose that $S_{r_1}$ is
  characteristic for $r_1$. In particular, $S_{r_1} \subseteq
  \lang(r_1)$. By definition, $M(S)$ is equivalent to $r$ for every
  sample $S$ with $S_{r_1} \subseteq S \subseteq \lang(r_1)$. We claim
  that in order for $M$ to have this property, $W$ must be a subset of
  $S_r$. Then, since $W$ contains all words over $B$ of length at most
  $n$, $|S_{r_1}| \geq \sum_{i=1}^{n}(n-2)^i$, as desired. The intuitive
  argument why $W$ must be a subset of $S_r$ is that if there exists
  $w$ in $W - S_r$, then $M$ cannot distinguish between $r_1$ and
  $r_w$. Indeed, suppose for the purpose of contradiction that there
  is some $w \in W$ with $w \not \in S_{r_1}$. Then $S_{r_1}$ is a
  subset of $\lang(r_w)$. Indeed, $S_{r_1} = S_{r_1} - \{w\} \subseteq
  \lang(r_1) - \{w\} = \lang(r_w)$. Furthermore, since $M$ learns the
  class of deterministic $(2n+3)\ores$, there must be some
  characteristic sample $S_{r_w}$ for $r_w$. Now, consider the sample
  $S_{r_1} \cup S_{r_w}$. It is included in both $\lang(r_1)$ and
  $\lang(r_w)$ and is a superset of both $S_{r_1}$ and $S_{r_w}$. But
  then, by definition of characteristic samples, $M(S_{r_1} \cup
  S_{r_w})$ must be equivalent to both $r_1$ and $r_w$. This is
  absurd, however, since $\lang(r_1) \not = \lang(r_w)$ by
  construction.

  A similar argument shows that the characteristic sample $S_{r_2}$ of
  $r_2 = (a_2 + \dots + a_n)^+a_1 (a_2 + \dots + a_n)^+$ also requires
  $\sum_{i=1}^{n}(n-2)^i$ elements. In this case, we take $B = A -
  \{a_1\}$ and we take $W$ to be the set of all non-empty words over
  $B$ of length at most $n$. For each $w = b_1 \dots b_m \in W$,
  we construct the deterministic $(2n+3)\ore$ $r_w$ such that
  $\lang(r_w)$ accepts all words in $\lang(r)$ that do not end with
  $a_1w$, as follows. Let, for $1 \leq i \leq m$, $r_w^i$ be the
  $2\ore$ that accepts all words in $B^+$ that do not start with
  $b_i$:
  \[ r_w^i := (B - \{b_i\}) \con B^* \]
  Then it suffices to take 
  \[r_w := B^+ a_1 (r_w^i + b_1(\emptystr +
  r_w^2 + b_3( \dots + b_{m-1}(\emptystr + r_w^m + b_mB^+) \dots ))).
 \]
 A similar argument as for $r_1$ then shows that the characteristic
 sample $S_{r_2}$ of $r_2$ needs to contain, for each $w \in W$, at
 least one word of the form $va_1w$ with $v \in B^+$. Therefore,
 $|S_{r_2}| \geq \sum_{i=1}^{n}(n-2)^i$, as desired.
\end{proof}


\section{The Learning  Algorithm}
\label{sec:mach-learn-appr}

In view of the observations made in Section~\ref{sec:basic-results},
we present in this section a practical learning algorithm that (1)
works well on incomplete data and (2) automatically determines the
best value of $k$ (see Section~\ref{sec:experiments} for an experimental
evaluation). Specifically, given a sample $S$, the algorithm derives
deterministic $\kores$ for increasing values of $k$ and selects from
these candidate expressions the \kore that describes $S$ best. To
determine the ``best'' expression we propose two measures: (1) a
Language Size measure and (2) a Minimum Description Length measure
based on the work of Adriaans and Vit\'anyi~\citeyear{Adri06a}.

Our algorithm does not derive deterministic $\kores$ for $S$ directly,
but uses, for each fixed $k$, a probabilistic method to first learn an
automaton for $S$, which is subsequently translated into a
$\kore$. The following section (Section \ref{sec:learn-koas-prob})
explains how the probabilistic method that learns an automaton from
$S$ works. Section~\ref{sec:transl-autom-into-kores} explains how the
learned automaton is translated into a \kore. Finally,
Section~\ref{sec:select-best-cand}, introduces the whole algorithm,
together with the two measures to determine the best candidate
expression.

\subsection{Probabilistically Learning a Deterministic Automaton}
\label{sec:learn-koas-prob}

In particular, the algorithm first learns a \emph{deterministic
  $k$-occurrence automaton} (deterministic \koa) for $S$. This is a
specific kind of finite state automaton in which each alphabet symbol
can occur at most $k$ times. Figure~\ref{fig:example-koa} gives an
example. Note that in contrast to the classical definition of an
automaton, no edges are labeled: all incoming edges in a state $s$ are
assumed to be labeled by the label of $s$.  In other words, the $2\oa$
of Figure~\ref{fig:example-koa} accepts the same language as $aa?b^+$.

\begin{definition}[(\koa)]
  An \emph{automaton} is a node-labeled graph $G = (V, E,
  \lab)$ where
  \begin{itemize}
  \item $V$ is a finite set of nodes (also called \emph{states}) with
    a distinguished source $\src \in V$ and sink $\sink \in V$;
  \item the edge relation $E$ is such that $\src$ has only outgoing
    edges; $\sink$ has only incoming edges; and every state $v \in V -
    \{\src,\sink\}$
    is reachable by a walk from $\src$ to $\sink$;
  \item $\lab\colon V - \{\src,\sink\} \to \alphabet$ is the labeling
    function.
  \end{itemize}
  In this context, an \emph{accepting run} for a word $a_1 \dots a_n$
  is a walk $\src s_1 \dots s_n \sink$ from $\src$ to $\sink$ in $G$
  such that $a_i = \lab(s_i)$ for $1 \leq i \leq n$. As usual, we
  denote by $\lang(G)$ the set of all words for which an accepting run
  exists.
  An automaton is \emph{$k$-occurrence} (a \koa) if there are at most
  $k$ states labeled by the same alphabet symbol. If $G$ uses only
  labels in $A \subseteq \alphabet$ then $G$ is \emph{an
  automaton over $A$}. \hspace{1cm} \qed
\end{definition}

In what follows, we write $\Succ(s)$ for the set $\{t \mid (s,t) \in
E\}$ of all direct successors of state $s$ in $G$, and $\Pred(s)$ for
the set $\{t \mid (t,s) \in E\}$ of all direct predecessors of $s$ in
$G$. Furthermore, we write $\Succ(s,a)$ and $\Pred(s,a)$ for the set
of states in $\Succ(s)$ and $\Pred(s)$, respectively, that are labeled
by $a$. As usual, an automaton $G$ is \emph{deterministic} if
$\Succ(s,a)$ contains at most one state, for every $s \in V$ and $a
\in \alphabet$.


For convenience, we will also refer to the $1$\oas as ``single
occurence automata'' or \soas for short.

We learn a deterministic $\koa$ for a sample $S$ as follows. First,
recall from Section~\ref{sec:basic-results} that $\alphabet(S)$ is the
set of alphabet symbols occurring in words in $S$. We view $S$ as the
result of a stochastic process that generates words from $\alphabet^*$
by performing random walks on the \emph{complete \koa $C_k$ over
  $\alphabet(S)$}.

\begin{definition}
  Define the \emph{complete $\koa$} $C_k$ over $\alphabet(S)$ to be
  the $\koa$ $G = (V,E,\lab)$ over $\alphabet(S)$ in which each $a \in
  \alphabet(S)$ labels exactly $k$ states such that
  \begin{itemize}
  \item there is an edge from $\src$ to $\sink$;
  \item $\src$ is connected to exactly one state labeled by $a$, for
    every $a \in \alphabet(S)$; and 
  \item every state $s \in V - \{\src,\sink\}$ has an outgoing edge to
    every other state except $\src$. \hspace{11cm} \qed
  \end{itemize}
\end{definition}

To illustrate, the complete $2\oa$ over $\{a,b\}$ is shown in
Figure~\ref{fig:example-complete}. Clearly, $\lang(C_k) = \alphabet(S)^*$.

\begin{figure}[t]
  \centering
  \subfigure[An example $2\oa$. It accepts the same language as $aa?b^+$]{
    \hspace{1cm}\includegraphics*[viewport=126 373 195 422]{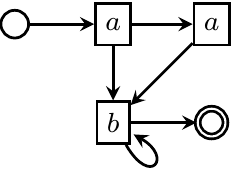}\hspace{1cm}
    \label{fig:example-koa}
  }
  \qquad\qquad
  \subfigure[The complete $2\oa$ over $\{a,b\}$.]
  {\includegraphics*[viewport=126 360 249 426]{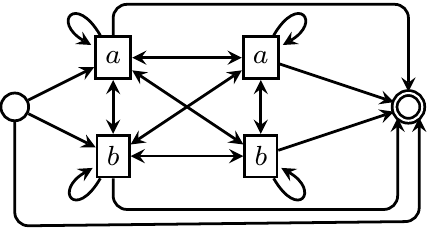}
    \label{fig:example-complete}
  }
  \caption{Two $2\oas$.}
\end{figure}

The stochastic process that generates words from $\alphabet^*$ by
performing random walks on $C_k$ operates as follows. First, the
process picks, among all states in $\Succ(\src)$, a state $s_1$ with
probability $\alpha(\src,s_1)$ and emits $\lab(s_1)$. Then it picks,
among all states in $\Out(s_1)$ a state $s_2$ with probability
$\alpha(s_1,s_2)$ and emits $\lab(s_2)$. The process continues moving
to new states and emitting their labels until the final state is
reached (which does not emit a symbol). Of course, $\alpha$ must be a
true probability distribution, i.e.,
\begin{equation}
  \label{eq:ml-1}
  \alpha(s,t) \geq 0; \quad \text{and} \quad \sum_{t \in \Out(s)}
  \alpha(s,t) = 1
\end{equation}
for all states $s \not = \sink$ and all states $t$.  The probability of generating a
particular accepting run $\vec{s} = \src s_1s_2\dots s_n\sink$ given
the process $\mathcal{P} = (C_k, \alpha)$ in this setting is
\[ \prob{\vec{s} \mid \mathcal{P}} = \alpha(\src,s_1) \cdot
\alpha(s_2,s_3) \cdot \alpha(s_2,s_3) \cdots
\alpha(s_n, \sink), \] and the probability of generating the
word $w = a_1 \dots a_n$ is
\[ \prob{w \mid \mathcal{P}} = \sum_{\text{all accepting runs
    $\vec{s}$ of $w$ in $C_k$}} \prob{\vec{s} \mid \mathcal{P}}. \]
Assuming independence, the probability of obtaining all words in the
sample $S$ is then
\[ \prob{S \mid \mathcal{P}} = \prod_{w \in S} \prob{w \mid
  \mathcal{P}}. \] Clearly, the process that best explains the
observation of $S$ is the one in which the probabilities $\alpha$ are
such that they maximize $\prob{S \mid \mathcal{P}}$.

To learn a deterministic $\koa$ for $S$ we therefore first try to
infer from $S$ the probability distribution $\alpha$ that maximizes
$\prob{S \mid \mathcal{P}}$, and use this distribution to determine
the topology of the desired deterministic $\koa$. In particular, we
remove from $C_k$ the non-deterministic edges with the lowest
probability as these are the least likely to contribute to the
generation of $S$, and are therefore the least likely to be necessary
for the acceptance of $S$.

The problem of inferring $\alpha$ from $S$ is well-studied in Machine
Learning, where our stochastic process $\mathcal{P}$ corresponds to a
particular kind of Hidden Markov Model sometimes referred to as a
Partially Observable Markov Model (\pomm for short). (For the readers
familiar with Hidden Markov Models we note that the initial state
distribution $\pi$ usually considered in Hidden Markov Models is
absorbed in the state transition distribution $\alpha(\src, \cdot)$ in
our context.) Inference of $\alpha$ is generally accomplished by the
well-known Baum-Welsh algorithm~\cite{Rabi89} that adjusts initial
values for $\alpha$ until a (possibly local) maximum is reached.

We use Baum-Welsh in our learning algorithm $\iKoa$ shown in
Algorithm~\ref{alg:iKoa}, which operates as follows. In line $1$,
$\iKoa$ initializes the stochastic process $\mathcal{P}$ to the tuple
$(C_k, \alpha)$ where
\begin{itemize}
\item $C_k$ is the complete \koa over $\alphabet(S)$;
\item $\alpha(\src,\sink)$ is the fraction of empty words in $S$;
\item $\alpha(\src,s)$ is the fraction of words in $S$ that start with
  $\lab(s)$, for every $s \in \Succ(\src)$; and
\item $\alpha(s,t)$ is chosen randomly for $s \not = \src$, subject to
  the constraints in equation \eqref{eq:ml-1}.
\end{itemize}
It is important to emphasize that, since we are trying to model a
stochastic process, multiple occurrences of the same word in $S$
\emph{are} important. A sample should therefore not be considered as a
set in Algorithm~\ref{alg:iKoa}, but as a \emph{bag}. Line $2$ then
optimizes the initial values of $\alpha$ using the Baum-Welsh algorithm.

\begin{algorithm}[t]
  \begin{algorithmic}[1]
    \REQUIRE a sample $S$, a value for $k$
    \ENSURE a deterministic \koa $G$ with $S \subseteq \lang(G)$
    \STATE $\mathcal{P} \leftarrow \Init(k,S)$
    \STATE $\mathcal{P} \leftarrow \BaumWelsh(\mathcal{P}, S)$
    \STATE $G \leftarrow \Disambiguate(\mathcal{P}, S)$
    \STATE $G \leftarrow \Prune(G, S)$
    \STATE \textbf{return} $G$
  \end{algorithmic}
  \caption{$\iKoa$}
  \label{alg:iKoa}
\end{algorithm}

\begin{algorithm}[tb]
  \begin{algorithmic}[1]
    \REQUIRE a POMM $\mathcal{P} = (G, \alpha)$ and sample $S$
    \ENSURE a deterministic $\koa$ 
    \STATE Initialize queue $Q$ to $\{ s \in \Succ(\src) \mid
    \alpha(\src,s) > 0\}$
    \STATE Initialize set of marked states $D \leftarrow \emptyset$
    \WHILE{$Q$ is non-empty}
     \STATE $s \leftarrow \First(Q)$
     \WHILE{some $a \in \Sigma$ has $\car{\Out(s,a)} > 1$}
       \STATE pick $t \in \Out(s,a)$ with $\alpha(s,t) =
       \max \{ \alpha(s,t') \mid t' \in \Out(s,a)\}$
       \STATE set $\alpha(s,t) \leftarrow \sum \{\alpha(s,t') \mid t' \in
         \Out(s,a)\}$
       \FOR{all $t'$ in $\Out(s,a) \setminus \{t\}$}
       \STATE delete edge $(s,t')$ from $G$
       \STATE set $\alpha(s,t') \leftarrow 0$
       \ENDFOR
       \STATE $\mathcal{P} \leftarrow \BaumWelsh(\mathcal{P},S)$
       \STATE \textbf{if} $S \not \subseteq \lang(G)$ \textbf{then Fail}
     \ENDWHILE
       \STATE add $s$ to marked states $D$ and pop $s$ from $Q$
       \STATE enqueue all states in $\Out(s) \setminus D$  to $Q$
   \ENDWHILE
   \STATE \textbf{return} $G$
  \end{algorithmic}
 \caption{\Disambiguate}
 \label{alg:disambiguate}
\end{algorithm}

With these probabilities in hand $\Disambiguate$, shown in
Algorithm~\ref{alg:disambiguate}, determines the topology of the
desired deterministic $\koa$ for $S$. In a breadth-first manner, it
picks for each state $s$ and each symbol $a$ the state $t \in
\Out(s,a)$ with the highest probability and deletes all other edges to
states labeled by $a$. Line $7$ merely ensures that $\alpha$ continues
to be a probability distribution after this removal and line $11$
adjusts $\alpha$ to the new topology. Line $12$ is a sanity check that
ensures that we have not removed edges necessary to accept all words
in $S$; $\Disambiguate$ reports failure otherwise. The result of a
successful run of $\Disambiguate$ is a deterministic $\koa$ which
nevertheless may have edges $(s,t)$ for which there is no
\emph{witness} in $S$ (i.e., a word in $S$ whose unique accepting run
traverses $(s,t)$). The function $\Prune$ in line $4$ of $\iKoa$
removes all such edges. It also removes all states $s \in \Succ(\src)$
without a witness in $S$.   Figure~\ref{fig:example}
illustrates a hypothetical run of $\iKoa$.

It should be noted that $\BaumWelsh$, which iteratively refines
$\alpha$ until a (possibly local) maximum is reached, is
computationally quite expensive. For that reason, our implementation
only executes a fixed number of refinement iterations of $\BaumWelsh$
in Line $11$. Rather surprisingly, this cut-off actually improves the
precision of $\learn$, as our experiments in
Section~\ref{sec:experiments} show, where it is discussed in more
detail.

\begin{figure*}[p]
{
\let\normalsize\scriptsize \normalsize

\subfigure[Process $\mathcal{P}$ returned by $\Init$ with random values for $\alpha$.]
{
\begin{tabular}{c}
\includegraphics*[viewport=125 360 249 426]{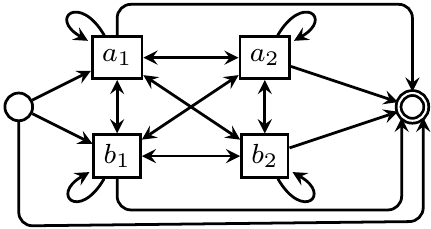}\\
    \begin{tabular}{r|ccccc}
   $\alpha$    & $a_1$ & $a_2$ & $b_1$ & $b_2$ & $\sink$ \\ \hline
$\src$  & 1 & $\backslash$ & 0 & $\backslash$ & 0 \\
$a_1$  & 0.2 & 0.3 & 0.3 & 0.1 & 0.1 \\
$a_2$  & 0.4 & 0.1 & 0.2 & 0.1 & 0.2 \\
$b_1$  & 0.1 & 0.3 & 0.3 & 0.2 & 0.1 \\
$b_2$  & 0.1 & 0.1 & 0.2 & 0.5 & 0.1 
  \end{tabular}
\end{tabular}
}
\hspace{0.3cm}
\subfigure[Process $\mathcal{P}$ after first training by $\BaumWelsh$.]
{
\begin{tabular}{c}
\includegraphics*[viewport=125 360 249 426]{fig/example-complete-2oa-indexed}\\
\begin{tabular}{r|ccccc}
    $\alpha$   & $a_1$ & $a_2$ & $b_1$ & $b_2$ & $\sink$ \\ \hline
$\src$  & 1 & $\backslash$ & 0 & $\backslash$ & 0 \\
$a_1$  & 0.2 & 0.3 & 0.3 & 0.19 & 0.01 \\
$a_2$  & 0.01 & 0.01 & 0.6 & 0.37 & 0.01 \\
$b_1$  & 0.01 & 0.01 & 0.5 & 0.28 & 0.2 \\
$b_2$  & 0.01 & 0.01 & 0.33 & 0.5 & 0.15 
  \end{tabular}
\end{tabular}
}
\vspace{1cm}
\subfigure[Process $\mathcal{P}$ after first disambiguation step (for
$a_1$). Edges to $a_1$ and $b_2$ are removed.]
{
\begin{tabular}{c}
\includegraphics*[viewport=125 360 249 426]{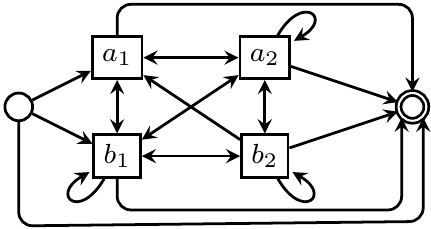}\\
  \begin{tabular}{r|ccccc}
    $\alpha$   & $a_1$ & $a_2$ & $b_1$ & $b_2$ & $\sink$ \\ \hline
$\src$  & 1 & $\backslash$ & 0 & $\backslash$ & 0 \\
$a_1$  & 0 & 0.5 & 0.49 & 0 & 0.01 \\
$a_2$  & 0.01 & 0.01 & 0.6 & 0.37 & 0.01 \\
$b_1$  & 0.01 & 0.01 & 0.5 & 0.28 & 0.2 \\
$b_2$  & 0.01 & 0.01 & 0.33 & 0.5 & 0.15 
  \end{tabular}
\end{tabular}
}
\hspace{0.3cm}
\subfigure[Process $\mathcal{P}$ after second disambiguation step (for
$b_1$). Edges to $a_2$ and $b_2$ are removed.]
{
\begin{tabular}{c}
\includegraphics*[viewport=125 360 249 426]{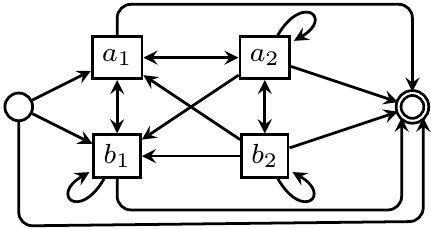}\\
  \begin{tabular}{r|ccccc}
    $\alpha$   & $a_1$ & $a_2$ & $b_1$ & $b_2$ & $\sink$ \\ \hline
$\src$  & 1 & $\backslash$ & 0 & $\backslash$ & 0 \\
$a_1$  & 0 & 0.5 & 0.49 & 0 & 0.01 \\
$a_2$  & 0.01 & 0.01 & 0.6 & 0.37 & 0.01 \\
$b_1$  & 0.02 & 0 & 0.78 & 0 & 0.2 \\
$b_2$  & 0.01 & 0.01 & 0.38 & 0.4 & 0.2 
  \end{tabular}
\end{tabular}
}
\subfigure[Automaton $A$ returned by $\Disambiguate$.]
{
  \hspace{0.7cm}
  \begin{tabular}{c}
\includegraphics*[viewport=126 360 249 426]{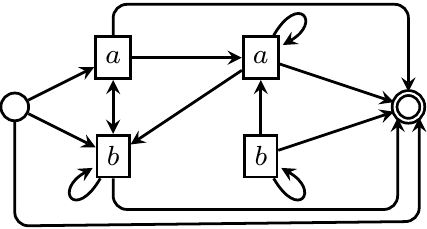}
  \end{tabular}
}
\hspace{1.5cm}
\subfigure[Automaton $A$ returned by $\Prune$. It accepts the same
language as $aa?b^+$.]
{
  \begin{tabular*}{5cm}{c}
  \hspace{1cm}
\includegraphics*[viewport=126 373 195 422]{fig/example-koa}
  \end{tabular*}
}
}
\caption{Example run of $\iKoa$ for $k=2$ with target language
  $aa?b^+$. For the process $\mathcal{P}$ in (c)-(f), the $\alpha$
  values are listed in table-form. To distinguish different states with the same
  label, we have indexed the labels.}
\label{fig:example}
\end{figure*}

\subsection{Translating \koas into \kores}
\label{sec:transl-autom-into-kores}


Once we have learned a deterministic \koa for a given sample $S$ using
\iKoa it remains to translate this \koa into a deterministic \kore. An
obvious approach in this respect would be to use the classical state
elimination algorithm (cf., e.g., \cite{Hopc79}). Unfortunately, as
already hinted upon by Fernau~\citeyear{Fern04,fernaualt} and as we
illustrate below, it is very difficult to get \emph{concise} regular
expressions from an automaton representation.  For instance, the
classical state elimination algorithm applied to the \soa in
Figure~\ref{fig:example-2tinf} yields the
expression:\footnote{Transformation computed by JFLAP:
  \url{www.jflap.org}.}
\[\small
\begin{array}{l}
  (aa^*d+(c+aa^*c)(c+aa^*c)^*(d+aa^*d)+(b+aa^*b+(c+{}\\
aa^*c)(c+aa^*c)^*(b+aa^*b))(aa^*b+(c+aa^*c)(c+aa^*c)^*\\
(b+aa^*b))^*(aa^*d+(c+aa^*c)(c+aa^*c)^*(d+aa^*d)))
(aa^*d+{}\\
(c+aa^*c)(c+aa^*c)^*(d+aa^*d)+(b+aa^*b+(c+aa^*c)(c+{}\\
aa^*c)^*(b+aa^*b))(aa^*b+(c+aa^*c)
(c+aa^*c)^*(b+aa^*b))^*\\
\end{array}
\]
which is non-deterministic and differs quite a bit from the equivalent
deterministic \sore $$\exre.$$

\begin{figure}[tbp]
  \centering
  \includegraphics*[viewport=125 363 224 426]{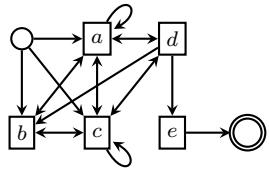}
  \caption{A \soa on which the classical state elimination algorithm returns a complicated expression.}
  \label{fig:example-2tinf}
\end{figure}

Actually, results by \citeN{Ehre76}; \citeN{stacs08}; and
\citeN{DBLP:conf/icalp/GruberH08} show that it is impossible in
general to generate concise regular expressions from automata: there
are \koas (even for $k = 1$) for which the number of occurrences of
alphabet symbols in the smallest equivalent expression is exponential
in the size of the automaton.  For such automata, an equivalent \kore
hence does not exist. 

It is then natural to ask whether there is an algorithm that
translates a given \koa into an equivalent \kore when such a \kore
exists, and returns a \kore super approximation of the input \koa
otherwise. Clearly, the above example shows that the classical state
elimination algorithm does not suffice for this purpose.  For that
reason, we have proposed in a companion article
\cite{gjb-sore-journal} a family of algorithms $\{ \cure, \kcure_1,
\kcure_2, \kcure_3, \dots \}$ that
translate \soas into \sores and have exactly these properties:

\begin{theorem}[(\cite{gjb-sore-journal})]
  \label{thm:www-kcure-sound}
  Let $G$ be a \soa and let $T$ be any of the algorithms in the family
  $\{\cure, \allowbreak \kcure_1, \allowbreak \kcure_2, \allowbreak
  \kcure_3, \dots\}$. If $G$ is equivalent to a \sore $r$, then $T(G)$
  returns a \sore equivalent to $r$. Otherwise, $T(G)$ returns a \sore
  that is a super approximation of $G$, $\lang(G) \subseteq
  \lang(T(G))$.
\end{theorem}

(Note that \soas and \sores are always deterministic by definition.)

These algorithms, in short, apply an inverse Glushkov
translation. Starting from a \koa where each state is labeled by a
symbol, they iteratively rewrite subautomata into equivalent regular
expressions. In the end only one state remains and the regular
expression labeling this state is the output.

In this section, we show how the above algorithms can be used to
translate \koas into \kores.  For simplicity of exposition, we will
focus our discussion on $\kcure_1$ as it is the concrete translation
algorithm used in our experiments in Section~\ref{sec:experiments},
but the same arguments apply to the other algorithms in the family.

\begin{definition}
  First, let $\mrk{\alphabet}{k}$ denote the alphabet that consists of $k$
  copies of the symbols in $\alphabet$, where the first copy of $a \in
  \alphabet$ is denoted by $\mrk{a}{1}$, the second by $\mrk{a}{2}$,
  and so on:
  \[ \mrk{\alphabet}{k} := \{ \mrk{a}{i} \mid a \in \alphabet, 1 \leq
  i \leq k \}. \] Let $\strip$ be the function mapping copies to
  their original symbol, i.e., $\strip(\mrk{a}{i}) = a$.  We extend $\strip$ pointwise
  to words, languages, and regular expressions over
  $\mrk{\alphabet}{k}$. \hspace{11.1cm} \qed
\end{definition}
For example, $\strip(\{ \mrk{a}{1}\mrk{a}{2}\mrk{b}{1}, \mrk{a}{2}
\mrk{a}{2}\mrk{c}{2}\}) = \{ aab, aac\}$ and $\strip(\mrk{a}{1} \con
\mrk{a}{2}?\con \allowbreak \mrk{b}{1}^+) = a\con a? \con b^+$\,.

To see how we can use $\kcure_1$, which translates \soas into \sores,
to translate a \koa into a \kore, observe that we can always transform
a \koa $G$ over $\alphabet$ into a \soa $H$ over $\mrk{\alphabet}{k}$
by processing the nodes of $G$ in an arbitrary order and replacing the
$i$th occurrence of label $a \in \alphabet$ by $\mrk{a}{i}$. To
illustrate, the \soa over $\mrk{\alphabet}{2}$ obtained in this way
from the $2\oa$ in Figure~\ref{fig:example-koa} is shown in
Figure~\ref{fig:example-marking}. Clearly, $\lang(G) = \strip(\lang(H))$.

\begin{definition}
  We call a \soa $H$ over $\mrk{\alphabet}{k}$ obtained from a \koa
  $G$ in the above manner a \emph{marking} of $G$. \hspace{8.2cm} \qed
\end{definition}

\begin{figure}[tbp]
  \centering
  \includegraphics*[viewport=125 371 198 424]{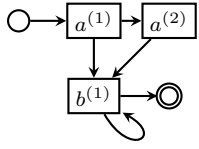}
  \caption{An example marking}
  \label{fig:example-marking}
\end{figure}

Note that, by Theorem~\ref{thm:www-kcure-sound}, running $\kcure_1$ on
$H$ yields a \sore $r$ over $\mrk{\alphabet}{k}$ with $\lang(H)
\subseteq \lang(r)$. For instance, with $H$ as in
Figure~\ref{fig:example-marking}, $\kcure(H)$ returns $r = \mrk{a}{1} \con
\mrk{a}{2}? \con \mrk{b}{1}^+$. By subsequently stripping $r$, we always
obtain a \kore over $\alphabet$. Moreover, $\lang(G) =
\strip(\lang(H)) \subseteq \strip(\lang(r)) = \lang(\strip(r))$, so
the \kore $\strip(r)$ is always a super approximation of
$G$. Algorithm~\ref{alg:KoaToKore}, called \KoaToKore, summarizes the
translation. By our discussion, \KoaToKore is clearly sound:

\begin{algorithm}[t]
  \begin{algorithmic}[1]
    \REQUIRE a \koa $G$
    \ENSURE a \kore $r$ with $\lang(G) \subseteq \lang(r)$
    \STATE compute a marking $H$ of $G$.
    \STATE \textbf{return} $\strip(\kcure_1(H))$
  \end{algorithmic}
  \caption{$\KoaToKore$}
  \label{alg:KoaToKore}
\end{algorithm}

\begin{proposition}
  \label{prop:koatokore-sound}
  $\KoaToKore(G)$ is a (possibly non-deterministic) \kore with
  $\lang(G) \subseteq \lang(\KoaToKore(G))$, for every $\koa$ $G$.
\end{proposition}

Note, however, that even when $G$ is deterministic and equivalent to a
deterministic \kore $r$, $\KoaToKore(G)$ need not be deterministic,
nor equivalent to $r$. For instance, consider the $2\oa$ $G$:
\[ \includegraphics*[viewport=125 373 221
  416]{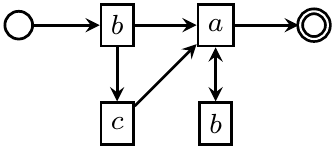} \] Clearly, $G$ is
equivalent to the deterministic $2\ore$ $bc?a(ba)^+?$.  Now suppose
for the purpose of illustration that $\KoaToKore$ constructs the
following marking $H$ of $G$. (It does not matter which marking
\KoaToKore constructs, they all result in the same final expression.)
\[ \includegraphics*[viewport=125 373 221
416]{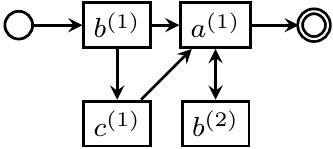} \] Since $H$ is not
equivalent to a \sore over $\mrk{\alphabet}{k}$, $\kcure_1(H)$ need
not be equivalent to $\lang(H)$. In fact, $\kcure_1(H)$ returns
$((\mrk{b}{1}\mrk{c}{1}?\mrk{a}{1})?\mrk{b}{2}?)^+$, which yields the
non-deterministic $((bc?a)?b?)^+$ after stripping. Nevertheless, $G$
is equivalent to the deterministic $2\ore$ $bc?a(ba)^+?$.

So although $\KoaToKore$ is always guaranteed to return a $\kore$, it
does not provide the same strong guarantees that $\kcure_1$ provides
(Theorem~\ref{thm:www-kcure-sound}). The following theorem shows,
however, that if we can obtain $G$ by applying the Glushkov
construction on $r$ \cite{Brug93}, $\KoaToKore(G)$ is always
equivalent to $r$. Moreover, if $r$ is deterministic, then so is
$\KoaToKore(G)$. So in this sense, $\KoaToKore$ applies an inverse
Glushkov construction to $r$. Formally, the Glushkov construction is
defined as follows.

\begin{definition}
  \label{def:glushkov-translation}
  Let $r$ be a \kore. Recall from Definition~\ref{def:deterministic}
  that $\overline{r}$ is the regular expression obtained from $r$ by
  replacing the $i$th occurrence of alphabet symbol $a$ by
  $\mrk{a}{i}$, for every $a \in \alphabet$ and every $1 \leq i \leq
  n$. Let $\pos(\overline{r})$ denote the symbols in
  $\mrk{\alphabet}{k}$ that actually appear in
  $\overline{r}$. Moreover, let the sets $\first(\overline{r})$,
  $\last(\overline{r})$, and $\follow(\overline{r},\mrk{a}{i})$ be
  defined as shown in Figure~\ref{fig:glushkov-sets}. A \koa $G$ is a
  \emph{Glushkov translation} of $r$ if there exists a one-to-one onto
  mapping $\rho \colon (V(G) - \{\src,\sink\}) \to \pos(\overline{r})$
  such that
  \begin{enumerate}
  \item $v \in \Succ(\src) \Leftrightarrow \rho(v) \in \first(\overline{r})$;
  \item $v \in \Pred(\sink) \Leftrightarrow \rho(v) \in \last(\overline{r})$;
  \item $v \in \Succ(w) \Leftrightarrow \rho(v) \in
    \follow(\overline{r},\rho(w))$; and
  \item $\strip(\rho(v)) = \lab(v)$,
  \end{enumerate}
  for all $v,w \in V(G) - \{\src,\sink\}$.\hspace{7cm} \qed
\end{definition}

\begin{figure}
  \centering
  \begin{tabular*}{1.0\linewidth}{rclcrcl}
    $\first(\emptyset)$ & = & $\emptyset$ & & $\first(\emptystr)$ & =
    & $\emptyset$ \\
    $\first(\mrk{a}{i})$ & = & \{\mrk{a}{i}\} & & $\first(\overline{r}?)$ & = & $\first(\overline{r})$\\
  $\first(\overline{r}^+)$ & = & $\first(\overline{r})$ & & $\first(\overline{r} + \overline{s})$ & = & $\first(\overline{r})
    \cup \first(\overline{s})$ \\
   $\first(\overline{r} \con \overline{s})$ & = & \multicolumn{5}{l}{$
     \begin{cases}
       \first(\overline{r}) & \text{if $\emptystr \notin \lang(\overline{r})$,}\\
       \first(\overline{r}) \cup \first(\overline{s})  & \text{otherwise.}\\
     \end{cases}$} \vspace{0.5cm}\\

    $\last(\emptyset)$ & = & $\emptyset$ & & $\last(\emptystr)$ & =
    & $\emptyset$ \\
    $\last(\mrk{a}{i})$ & = & \{\mrk{a}{i}\} & & $\last(\overline{r}?)$ & = & $\last(\overline{r})$\\
  $\last(\overline{r}^+)$ & = & $\last(\overline{r})$ & & $\last(\overline{r} + \overline{s})$ & = & $\last(\overline{r})
    \cup \last(\overline{s})$ \\
   $\last(\overline{r} \con \overline{s})$ & = & \multicolumn{5}{l}{
     $\begin{cases}
       \last(\overline{s}) & \text{if $\emptystr \notin \lang(\overline{s})$,}\\
       \last(\overline{r}) \cup \last(\overline{s})  & \text{otherwise.}\\
     \end{cases}$}\vspace{0.5cm} \\

    $\follow(\mrk{a}{i},\mrk{a}{i})$ & = & $\emptyset$ & &
     &  & \\

    $\follow(\overline{r}?,\mrk{a}{i})$ & =  &  $\follow(\overline{r},\mrk{a}{i})$ & &
     &  &\\

  $\follow(\overline{r}^+,\mrk{a}{i})$ & = & 
  \multicolumn{5}{l}{$\begin{cases}
       \follow(\overline{r},\mrk{a}{i}) & \hspace{17pt} \mbox{if } \mrk{a}{i} \notin \last(\overline{r}),\\
       \follow(\overline{r},\mrk{a}{i}) \cup \first(\overline{r})  &
       \hspace{17pt} \text{otherwise.}\\
     \end{cases}$}\\

  $\follow(\overline{r} + \overline{s},\mrk{a}{i})$ & = & 
  \multicolumn{5}{l}{$\begin{cases}
       \follow(\overline{r},\mrk{a}{i}) & \hspace{36pt} \mbox{if } \mrk{a}{i} \in
       \pos(\overline{r}),\\
       \follow(\overline{s},\mrk{a}{i})  & \hspace{36pt} \text{otherwise.}\\
     \end{cases}$} \\

  $\follow(\overline{r} \con \overline{s},\mrk{a}{i})$ & = & 
  \multicolumn{5}{l}{$\begin{cases}
       \follow(\overline{r},\mrk{a}{i}) & \mbox{if } \mrk{a}{i} \in
       \pos(\overline{r}), \mrk{a}{i} \notin \last(\overline{r}),\\
       \follow(\overline{r},\mrk{a}{i}) \cup \first(\overline{s})& \mbox{if } \mrk{a}{i} \in
       \pos(\overline{r}), \mrk{a}{i} \in \last(\overline{r}),\\
\follow(\overline{s},\mrk{a}{i}) & \text{otherwise.}\\
     \end{cases}$}\\
  \end{tabular*}
  \caption{Definition of $\first(\overline{r})$,
    $\last(\overline{r})$, and $\follow(\overline{r},\mrk{a}{i})$, for
    $\mrk{a}{i} \in \pos(\overline{r})$.}
  \label{fig:glushkov-sets}
\end{figure}

\begin{theorem}
  \label{thm:koatokore-complete}
  If \koa $G$ is a Glushkov representation of a target
  $\kore$ $r$, then $\KoaToKore(G)$ is equivalent to $r$. Moreover, if
  $r$ is deterministic, then so is $\KoaToKore(G)$.
\end{theorem}
\begin{proof}
  Since $\KoaToKore(G) = \strip(\kcure_1(H))$ for an arbitrarily
  chosen marking $H$ of $G$, it suffices to prove that
  $\strip(\kcure_1(H))$ is equivalent to $r$ and that
  $\strip(\kcure_1(H))$ is deterministic whenever $r$ is
  deterministic, for every marking $H$ of $G$.  Hereto, let $H$ be an
  arbitrary but fixed marking of $G$. In particular, $G$ and $H$ have
  the same set of nodes $V$ and edges $E$, but differ in their
  labeling function. Let $\lab_G$ be the labeling function of $G$ and
  let $\lab_H$ the labeling function of $H$. Clearly, $\lab_G(v) =
  \strip(\lab_H(v))$ for every $v \in V - \{\src,\sink\}$.  Since $G$
  is a Glushkov translation of $r$, there is a one-to-one, onto
  mapping $\rho \colon (V - \{\src,\sink\})\to \pos(\overline{r})$
  satisfying properties (1)-(4) in
  Definition~\ref{def:glushkov-translation}.  Now let $\sigma\colon
  \pos(\overline{r}) \to \mrk{\alphabet}{k}$ be the function that maps
  $\mrk{a}{i} \in \pos(\overline{r})$ to
  $\lab_H(\rho^{-1}(\mrk{a}{i}))$. Since $\lab_H$ assigns a distinct
  label to each state, $\sigma$ is one-to-one and onto the subset of
  $\mrk{\alphabet}{k}$ symbols used as labels in $H$. Moreover, by
  property (4) and the fact that $\lab_G(v) = \strip(\lab_H(v))$ we
  have,
  \begin{equation}
    \label{thm:koatokore-complete-eq-1}
    \strip(\mrk{a}{i}) = \lab_G(\rho^{-1}(\mrk{a}{i})) =
  \strip(\lab_H(\rho^{-1}(\mrk{a}{i}))) =
  \strip(\sigma(\mrk{a}{i})) \tag{$\star$}
  \end{equation}
 for each $\mrk{a}{i} \in
  \pos(\overline{r})$. In other words, $\sigma$ preserves (stripped)
  labels. Now let $\sigma(\overline{r})$ be the \sore obtained from
  $\overline{r}$ by replacing each $\mrk{a}{i} \in \pos(\overline{r})$ by
  $\sigma(\mrk{a}{i})$. Since $\sigma$ is one-to-one and
  $\overline{r}$ is a \sore, so is $\sigma(\overline{r})$. Moreover,
  we claim that $\lang(H) = \lang(\sigma(\overline{r}))$. 

  Indeed, it is readily verified by induction on $\overline{r}$ that a
  word $\mrk{a_1}{i_1} \dots \mrk{a_n}{i_n} \in \lang(\overline{r})$
  if, and only if, (i) $\mrk{a_1}{i_1} \in \first(\overline{r})$; (ii)
  $\mrk{a_{p+1}}{i_{p+1}} \in \follow(\overline{r},
  \mrk{a_{p+1}}{i_{p+1}})$ for every $1\leq p < n$; and (iii)
  $\mrk{a_n}{i_n} \in \last(\overline{r})$. By properties (1)-(4) of
  Definition~\ref{def:glushkov-translation} we hence obtain:
  \[ 
  \begin{array}{ll}
    & \sigma(\mrk{a_1}{i_1}) \dots \sigma(\mrk{a_n}{i_n}) \in
    \lang(\sigma(\overline{r})) \\
    \Leftrightarrow & 
    \mrk{a_1}{i_1} \dots \mrk{a_n}{i_n} \in \lang(\overline{r}) \\
    \Leftrightarrow & 
    \src, \rho^{-1}(\mrk{a_1}{i_1}), \dots, \rho^{-1}(\mrk{a_n}{i_n}), \sink
    \text{ is a walk in } G \\
    \Leftrightarrow & 
    \src, \rho^{-1}(\mrk{a_1}{i_1}), \dots, \rho^{-1}(\mrk{a_n}{i_n}), \sink
    \text{ is a walk in } H \\
    \Leftrightarrow & 
    \lab_H(\rho^{-1}(\mrk{a_1}{i_1})) \dots,
    \lab_H(\rho^{-1}(\mrk{a_n}{i_n})) \in \lang(H) \\
    \Leftrightarrow & 
     \sigma(\mrk{a_1}{i_1}) \dots \sigma(\mrk{a_n}{i_n}) \in
    \lang(H)
  \end{array}\]
  Therefore, $\lang(H) = \lang(\sigma(\overline{r}))$. 

  Hence, we have established that $H$ is a \soa over
  $\mrk{\alphabet}{k}$ equivalent to the \sore $\sigma(\overline{r})$
  over $\mrk{\alphabet}{k}$. By Theorem~\ref{thm:www-kcure-sound},
  $\kcure_1(H)$ is hence equivalent to
  $\sigma(\overline{r})$. Therefore, $\strip(\kcure_1(H))$ is
  equivalent to $\strip(\sigma(\overline{r}))$, which by
  \eqref{thm:koatokore-complete-eq-1}  
  above, is equivalent to $\strip(\overline{r}) = r$, as desired.

  Finally, to see that $\strip(\kcure_1(H))$ is deterministic if $r$
  is deterministic, let $s := \strip(\kcure_1(H))$ and suppose for the
  purpose of contradiction that $s$ is not deterministic. Then there
  exists $w\mrk{a}{i}v_1$ and $w\mrk{a}{j}{v_2}$ in $\lang(\overline{s})$
  with $i \not = j$. It is not hard to see that this can happen only
  if there exist $w'\mrk{a}{i'}v_1'$ and $w'\mrk{a}{j'}{v_2'}$ in
  $\lang(\kcure_1(H))$ with $i' \not = j'$.  Since $\lang(\kcure_1(H))
  = \lang(\sigma(\overline{r}))$ we know that hence
  $\sigma^{-1}(w'\mrk{a}{i'}v_1') \in \lang(\overline{r})$ and
  $\sigma^{-1}(w'\mrk{a}{j'}v_2') \in \lang(\overline{r})$. Let
  $w''\mrk{a}{i''}v_1'' = \sigma^{-1}(w'\mrk{a}{i'}v_1')$ and
  $w''\mrk{a}{j''}v_2'' = \sigma^{-1}(w'\mrk{a}{i'}v_2')$. Since $\sigma$
  is one-to-one and $i' \not = j'$, also $i'' \not = j''$. Therefore,
  $r$ is not deterministic, which yields the desired contradiction.
\end{proof}

\subsection{The whole Algorithm}
\label{sec:select-best-cand}

Our deterministic regular expression inference algorithm $\learn$
combines $\iKoa$ and $\KoaToKore$ as shown in
Algorithm~\ref{alg:Learn}. For increasing values of $k$ until a
maximum $k_{\mathrm{max}}$ is reached, it first learns a deterministic
\koa $G$ from the given sample $S$, and subsequently translates that
\koa into a \kore using $\KoaToKore$. If the resulting \kore is
deterministic then it is added to the set $C$ of deterministic
candidate expressions for $S$, otherwise it is discarded.  From this
set of candidate expressions, $\learn$ returns the ``best'' regular
expression $\Best(C)$, which is determined according to one of the
measures introduced below. Since it is well-known that, depending on
the initial value of $\alpha$, $\BaumWelsh$ (and therefore $\iKoa$)
may converge to a local maximum that is not necessarily global, we
apply $\iKoa$ a number of times $N$ with independently chosen random
seed values for $\alpha$ to increase the probability of correctly
learning the target regular expression from $S$.

\begin{algorithm}[tbp]
  \begin{algorithmic}[1]
    \REQUIRE a sample $S$
    \ENSURE a $\kore$ $r$
    \STATE initialize candidate set $C \leftarrow \emptyset$
    \FOR{$k = 1$ to $k_{\mathrm{max}}$}
      \FOR{$n = 1$ to $N$}
      \STATE $G \leftarrow \iKoa(S,k)$
      \IF{$\KoaToKore(G)$ is deterministic}
      \STATE add $\KoaToKore(G)$ to $C$
      \ENDIF
      \ENDFOR
    \ENDFOR
    \STATE \textbf{return} $\Best(C)$  
  \end{algorithmic}
  \caption{\learn}
  \label{alg:Learn}
\end{algorithm}

The observant reader may wonder whether we are always guaranteed to
derive at least one deterministic expression such that $\Best(C)$ is
defined.  Indeed, Theorem~\ref{thm:koatokore-complete} tells us that
if we manage to learn from sample $S$ a $\koa$ which is the Glushkov
representation of the target expression $r$, then $\KoaToKore$ will
always return a deterministic $\kore$ equivalent to $r$.  When $k>1$,
there can be several $\koa$s representing the same language and we
could therefore learn a non-Glushkov one. In that case, $\KoaToKore$
always returns a $\kore$ which is a super approximation of the target
expression. Although that approximation can be non-deterministic,
since we derive $\kores$ for increasing values of $k$ and since for
$k=1$ the result of $\KoaToKore$ is always deterministic (as every
$\sore$ is deterministic), we always infer at least one deterministic
regular expression.  In fact, in our experiments on 100 synthetic
regular expressions, we derived for 96 of them a deterministic
expression with $k>1$, and only for 4 expressions had to resort to a
$1\ore$ approximation.

\subsubsection{A Language Size Measure for Determining the Best
  Candidate}
\label{sec:language-size-select-best-cand}


Intuitively, we want to select from $C$ the simplest deterministic
expression that ``best'' describes $S$. Since each candidate
expression in $C$ accepts all words in $S$ by construction, one way to
interpret ``the best'' is to select the expression that accepts the
least number of words (thereby adding the least number of words to
$S$). Since an expression defines an infinite language in general, it
is of course impossible to take all words into account. We therefore
only consider the words up to a length $n$, where $n = 2m + 1$ with
$m$ the length of the candidate expression, excluding regular
expression operators, $\emptyset$, and $\emptystr$. For instance, if
the candidate expression is $a \con (a + c^+)?$, then $m = 3$ and $n =
7$. Formally, for a language $L$, let $|L^{\leq n}|$ denote the number
of words in $L$ of length at most $n$. Then the best candidate in $C$
is the one with the least value of $|\lang(r)^{\leq n}|$. If there are
multiple such candidates, we pick the shortest one (breaking ties
arbitrarily). It turns out that $|\lang(r)^{\leq n}|$ can be computed
quite efficiently; see \cite{gjb-sore-journal} for details.

\subsubsection{A Minimum Description Length Measure for Determining
  the Best Candidate}
\label{sec:mdl-select-best-cand}

An alternative measure to determine the best candidate is given by
Adriaans and Vit\'anyi~\citeyear{Adri06a}, who compare the size of $S$
with the size of the language of a candidate $r$.  Specifically,
Adriaans and Vit\'anyi define the data encoding cost of $r$ to be:
\[ \Data(r, \Corpus) := \sum_{i = 0}^{n} \left( 2 \cdot \log_2 i +
  \log_2 {|\lang^{=i}(r)| \choose |\Corpus^{=i}|} \right), \] where $n
= 2m + 1$ as before; $|S^{=i}|$ is the number of words in $S$ that
have length $i$; and $|\lang^{=i}(r)|$ is the number of words in
$\lang(r)$ that have exactly length $i$. Although the above formula is
numerically difficult to compute, there is an easier estimation
procedure; see~\cite{Adri06a} for details. 

In this case, the model encoding cost is simply taken to be its
length, thereby preferring shorter expressions over longer ones. The
best regular expression in the candidate set $C$ is then the one that
minimizes both model and data encoding cost (breaking ties
arbitrarily).

We already mentioned that \xtract~\cite{Garo03} also utilizes the
Minimum Description Length principle. However, their measure for data
encoding cost depends on the concrete structure of the regular
expressions while ours only depends on the language defined by them
and is independent of the representation. Therefore, in our setting,
when two equivalent expressions are derived, the one with the smallest
model cost, that is, the simplest one, will always be taken.



\section{Experiments}
\label{sec:experiments}

In this section we validate our approach by means of an experimental
analysis. Throughout the section, we say that a target \kore $r$
\emph{is successfully derived} when a \kore $s$ with $\lang(r) =
\lang(s)$ is generated. The \emph{success rate} of our experiments
then is the percentage of successfully derived target regular
expressions.

Our previous work~\cite{Bex08} on this topic was based on a version of
the \vldb algorithm~\cite{Bex06}, we refer to this algorithm as
$\learn(\vldb)$.  Unfortunately, as detailed in~\cite{vldbj}, it is
not known whether \vldb is complete on the class of all single
occurrence regular expressions.  Nevertheless, the experiments
in~\cite{Bex08} which are revisited below show a good and reliable
performance.  However, to obtain a theoretically complete algorithm,
c.f.r.~Theorem~\ref{thm:koatokore-complete}, we use the algorithm
\kcure which is sound and complete on single occurrence regular
expressions.  In the remainder we focus on \learn, but compare with
the results for $\learn(\vldb)$.

As mentioned in Section~\ref{sec:language-size-select-best-cand},
another new aspect of the results presented here is the use of
language size as an alternative measure over Minimum Description
Length (MDL) to compare
candidates.  The $\learn(\vldb)$ algorithm is only considered with
the MDL criterion.  We note that for alphabet size 5, the success rate
of \learn with the MDL criterion was only 21 \%, while that of the
language size criterion is 98 \%.  The corpus used in this experiment
is described in Section~\ref{sec:synth-target}.  Therefore in the
remainder of this section we only consider \learn with the language
size criterion.

For all the experiments described below we take $k_{\text{max}}=4$ and
$N=10$ in Algorithm~\ref{alg:Learn}.

\subsection{Running times}
\label{sec:runtimes}

All experiments were performed using a prototype implementation of
$\learn$ and $\learn(\vldb)$ written in Java executed on Pentium M 2.0
GHz class machines equipped with 1GB RAM. For the $\BaumWelsh$
subroutine we have gratefully used Jean-Marc Fran\c{c}ois'
\emph{Jahmm} library~\cite{Jahmm}, which is a faithful implementation
of the algorithms described in Rabiner's Hidden Markov Model
tutorial~\cite{Rabi89}. Since Jahmm strives for clarity rather than
performance and since only limited precautions are taken against
underflows, our prototype should be seen as a proof of concept rather
than a polished product. In particular, underflows currently limit us
to target regular expressions whose total number of symbol occurrences
is at most $40$. Here, the total number of symbol occurrences
$\occ(r)$ of a regular expression $r$ is its length excluding the
regular expression operators and parenthesis. To illustrate, the total
number of symbol occurrences in $aa?b^+$ is $3$. Furthermore, the lack
of optimization in Jahmm leads to average running times ranging from 4
minutes for target expressions $r$ with $|\Sigma(r)| = 5$ and $\occ(r)
= 6$ to 9 hours for targets expression with $|\Sigma(r)| = 15$ and
$\occ(r) = 30$.  Running times for \learn and $\learn(\vldb)$ are
similar.

As already mentioned in Section~\ref{sec:select-best-cand}, one of the
bottlenecks of $\learn$ is the application of \BaumWelsh in Line $11$
of $\Disambiguate$ (Algorithm~\ref{alg:disambiguate}).  \BaumWelsh is
an iterative procedure that is typically run until convergence, i.e.,
until the computed probability distribution no longer change
significantly.  To improve the running time, we only apply a fixed
number $\ell$ of iteration steps when calling $\BaumWelsh$ in Line
$11$ of $\Disambiguate$.  Experiments show that the running time
performance scales linear with $\ell$ as one expects, but, perhaps
surprisingly, the success rate improves as well for an optimal value
of $\ell$.  This optimal value for $\ell$
depends on the alphabet size.  These improved results can be explained
as follows: applying \BaumWelsh in each disambiguation step until it
converges guarantees that the probability distribution for that step
will have reached a local optimum.  However, we know that the search
space for the algorithm contains many local optima, and that
\BaumWelsh is a local optimization algorithm, i.e., it will converge
to one of the local optima it can reach from its starting point by
hill climbing.  The disambiguation procedure proceeds state by state,
so fine tuning the probability distribution for a disambiguation step
may transform the search space so that certain local optima for the
next iteration can no longer be reached by a local search algorithm
such as \BaumWelsh.  Table~\ref{tab:lim-disamb} shows the performance
of the algorithm for various number of \BaumWelsh iterations $\ell$
for expressions of alphabet size 5, 10 and 15.  These expressions are
those described in Section~\ref{sec:synth-target}.  In this Table,
$\ell = \infty$ denotes the case where \BaumWelsh is ran until
convergence after each disambiguation step.  The Table illustrates
that the success rate is actually higher for small values of $\ell$.
The running time performance gains increase rapidly with the
expressions' alphabet size: for $|\alphabet| = 5$, we gain a factor of
3.5 ($\ell = 2$), for $|\alphabet| = 10$, it is already a factor of 10
($\ell = 3$) and for $|\alphabet| = 15$, we gain a factor of 25 ($\ell
= 3$).  This brings the running time for the largest expressions we
tested down to 22 minutes, in contrast with 9 hours mentioned for
$\learn(\vldb)$ and \learn.  The algorithm with the optimal number of
\BaumWelsh steps in the disambiguation process will be referred to as
$\fixed$.  In particular for small alphabet sizes ($|\alphabet| \le 7$)
we use $\ell = 2$, for large alphabet size $\ell = 3$ ($|\alphabet| >
7$).  We note that the alphabet size can easily be determined from the
sample.

We should also note that Experience with Hidden Markov Model learning
in bio-informatics~\cite{Finn06} suggests that both the running time
and the maximum number of symbol occurrences that can be handled can
be significantly improved by moving to an industrial-strength
$\BaumWelsh$ implementation. Our focus for the rest of the section
will therefore be on the precision of $\learn$.

\begin{table}[tbh]
  \centering
  \begin{tabular}{r|r|r|r}
    $\ell$ & rate $|\alphabet| = 5$ & rate $|\alphabet| = 10$ & rate $|\alphabet| = 15$ \\ \hline
    1 &       95 \%     & 80 \%          & 40 \% \\
    2 & \textbf{100 \%} & 75 \%          & 50 \% \\
    3 &       95 \%     & \textbf{84 \%} & \textbf{60 \%} \\
    4 &       95 \%     & 77 \%          & 50 \% \\ \hline
    $\infty$ & 98 \%    & 75 \%          & 50 \%
  \end{tabular}
  \caption{Success rate for a limited number of \BaumWelsh iterations in the
    disambiguation procedure, $\ell = \infty$ corresponds to \learn, for $\ell = 1,\dots,4$ correspond to $\fixed$.}
  \label{tab:lim-disamb}
\end{table}

\subsection{Real-world target expressions and real-world samples}
\label{sec:real-world-target}

We want to test how $\learn$ performs on real-world data. Since the
number of publicly available XML corpora with valid schemas is rather
limited, we have used as target expressions the $49$ content models
occurring in the \xsd for XML Schema Definitions~\cite{XSDS01} and
have drawn multiset samples for these expressions from a large corpus
of real-world \xsds harvested from the Cover Pages~\cite{Cover03}. In
other words, the goal of our first experiment is to derive, from a
corpus of \xsd definitions, the regular expression content models in
the schema for XML Schema Definitions\footnote{This corpus was also
  used in \cite{Bex07} for XSD inference.}.  As it turns out, the \xsd
regular expressions are all single occurrence regular expressions.

The $\learn(\vldb)$ algorithm infers all these expressions correctly,
showing that it is conservative with respect to $k$ since, as
mentioned above, the algorithm considers $k$ values ranging from 1 to
4.  In this setting, \learn performs not as well, deriving only 73 \%
of the regular expressions correctly.  We note that for each
expression that was not derived exactly, always an expression was
obtained describing the input sample and which in addition is more
specific than the target expression.  \learn therefore seems to favor
more specific regular expressions, based on the available examples.


\subsection{Synthetic target expressions}
\label{sec:synth-target}

Although the successful inference of the real-world expressions in
Section~\ref{sec:real-world-target} suggests that $\learn$ is
applicable in real-world scenarios, we further test its behavior on a
sizable and diverse set of regular expressions.  Due to the lack of
real-world data, we have developed a synthetic regular expression
generator that is parameterized for flexibility.

\myparagraph{Synthetic expression generation} In particular, the
occurrence of the regular expression operators concatenation,
disjunction ($+$), zero-or-one ($?$), zero-or-more ($^*$), and
one-or-more ($^+$) in the generated expressions is determined by a
user-defined probability distribution.  We found that typical values
yielding realistic expressions are $1/10$ for the unary operators and
$7/20$ for others.  The alphabet can be specified, as well as the
number of times that each individual symbol should occur.  The maximum
of these numbers determines the value $k$ of the generated \kore.

To ensure the validity of our experiments, we want to generate a wide
range of different expressions. To this end, we measure how much the
language of a generated expression overlaps with $\Sigma^*$. The
larger the overlap, the greater its language size as defined in
Section~\ref{sec:language-size-select-best-cand}.

To ensure that the generated expressions do not impede readability by
containing redundant subexpressions (as in e.g., $(a^{+})^{+}$), the
final step of our generator is to syntactically simplify the generated
expressions using the following straightforward equivalences:
\begin{eqnarray*}
  r^{*}                     & \rightarrow & r^{+}? \\
  r??                      & \rightarrow & r? \\
  (r^{+})^{+}                & \rightarrow & r^{+} \\
  (r?)^{+}                  & \rightarrow & r^{+}? \\
  (r_1 \cdot r_2) \cdot r_3 & \rightarrow & r_1 \cdot (r_2 \cdot r_3) \\
  r_1 \cdot (r_2 \cdot r_3) & \rightarrow & r_1 \cdot r_2 \cdot r_3 \\
  (r_1? \cdot r_2?)?        & \rightarrow & r_1? \cdot r_2? \\
  (r_1 + r_2) + r_3         & \rightarrow & r_1 + (r_2 + r_3) \\
  r_1 + (r_2 + r_3)         & \rightarrow & r_1 + r_2 + r_3 \\
  (r_1 + r_2^{+})^{+}        & \rightarrow & (r_1 + r_2)^{+} \\
  (r_1^{+} + r_2^{+})        & \rightarrow & (r_1 + r_2)^{+} \\
  r_1 + r_2? & \rightarrow  & (r_1 + r_2)?
\end{eqnarray*}
Of course, the resulting expression is rejected if it is
non-deterministic.

To obtain a diverse target set, we synthesized expressions with
alphabet size $5$ ($45$ expressions), $10$ ($45$ expressions), and
$15$ ($10$ expressions) with a variety of symbol occurrences ($k = 1,
2, 3$). For each of the alphabet sizes, the expressions were selected
to cover language size ranging from $0$ to $1$.  All in all, this
yielded a set of $100$ deterministic target expressions. A snapshot is
given in Figure~\ref{fig:snap}.

\begin{figure*}
\scriptsize
\centering
\begin{tabular}{l}
    $((d e b a b) + c)^{*} a$    \\
    $((((c + b) b) + a) c a) + e + d$    \\
    $(((e a)^{*} d b) + b + a + c)^{+}$    \\
    $((b^{+} + c + e + d) a a b)^{+}$    \\
    $((((e a b h) + d + j + c + b)^{+} f) + a + g + i)?$    \\
    $((((a a) + e)^{+} + c) b) + b + d$   \\
    $((((d + a)^{*} e a b c b) + c) a)?$   \\
    $((((a c) + b + d) e a b) + c)^{*}$   \\
    $(((((b a b) + c)^{+} + e)? a) + d)^{+}$   \\
    $((((e c b)^{+} a) + b)^{+} + d + a)?$    \\
    $((b a g b f e i d) + c + a + j + h)^{*}$    \\
    $((g d a b) + a + i + c + j + e + f)^{+} h b$    \\
    $((h^{*} c d f a) + j + e + g + b + i)^{*} a b$    \\
    $((g + b + e + f + i + d)^{*} a b a) + h + j + c$    \\
    $((((h + b + c + j + f)^{+} + e)? a a i d b) + g)?$ \\
  \end{tabular}
  \begin{tabular}{l}
    $(((((d b e)^{*} c f) + j) h a c) + b + i)^{*} g a d$    \\
    $(((((i h a a j) + d)^{+} + g) b) + e + b + f + c)^{+}$    \\
    $(((e c g e c d) + b + d + a + j + f)^{*} i h a b a)^{*}$    \\
    $(l + c + d + m + n)^{*} a o j a h b e g c b f i d k e$    \\
    $(((c + b) a b) + d + i + a)^{+} + j + g + f + e + h$    \\
    $(((a? c l f h a b g d) + b + n + o) i e d j c e m)^{*} k$    \\
    $((a + k + f + c + m + e)^{+} b d i e c l b o n j g d a)^{*} h$    \\
    $(((k? j g h a d f c e l i f c j b h o m) +$ \\
    \hspace{2cm}$b + g + a + e + i + n)^{+} + d)?$    \\
    $(((a e d o a d e n h d b c i) + h + k + m + j + g + b)^{*}$ \\
    \hspace{4.5cm}$f c c g e l b i f j a)$    \\
    $((a^{+} + f + d + o + g + n + h + c + b + j + i + e)$\\
    \hspace{4.7cm}$k e a c d l b m)$    \\
    $(((k + f + o + a + j)? e d h l d f h n g i c j m a b)? c i e)^{*}
    b g$    \\
    $((((a? d)^{+} b a) + h + g + e + c)^{+} + j + i + b)? f$    \\
  \end{tabular}


\caption{\label{fig:snap} A snapshot of the 100 generated expressions.}
\end{figure*}

\myparagraph{Synthetic sample generation} For each of those $100$
target expressions, we generated synthetic samples by transforming the
target expressions into stochastic processes that perform random walks
on the automata representing the expressions (cf.
Section~\ref{sec:mach-learn-appr}). The probability distributions of
these processes are derived from the structure of the originating
expression.  In particular, each operand in a disjunction is equally
likely and the probability to have zero or one occurrences for the
zero-or-one operator $?$ is $1/2$ for each option. The probability to
have $n$ repetitions in a one-or-more or zero-or-more operator ($^*$
and $^+$) is determined by the probability that we choose to continue
looping ($2/3$) or choose to leave the loop ($1/3$).  The latter
values are based on observations of real-world corpora.
Figure~\ref{fig:graph-rewriting} illustrates how we construct the
desired stochastic process from a regular expression $r$: starting
from the following initial graph,
\begin{center}
  \scriptsize
  \begin{tikzpicture}[>=stealth]
    \tikzstyle{every path}+=[thick]
    \tikzstyle{init}=[circle,draw,inner sep=2pt]
    \tikzstyle{mystate}=[transform shape,rectangle,thick,draw,minimum height=3.5ex]
    \tikzstyle{final}=[circle,draw,double,inner sep=2pt,thick]

    \node[init] (init)  at (0,0)   {\phantom{$r$} };
    \node[mystate] (r)     at (1,0)   {$r$};
    \node[final] (final) at (2,0)   {\phantom{$r$} };


    \draw[->] (init) -- node[above] {\scriptsize $1$} (r);
    \draw[->] (r) -- node[above] {\scriptsize $1$} (final);    
  \end{tikzpicture}
\end{center}  
we continue applying the rewrite rules shown until each internal node
is an individual alphabet symbol.

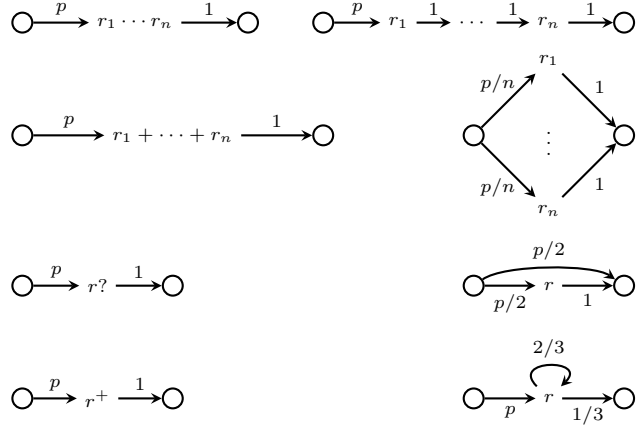
\begin{figure}[t]
  \centering
  \scriptsize
  \begin{tikzpicture}[>=stealth]
    \tikzstyle{every path}+=[thick]
    \tikzstyle{init}=[circle,draw,thick]
    \tikzstyle{final}=[circle,draw,thick]
    \def\concaty{5}
    \def\optionaly{1.5}
    \def\plusy{0}

    \node[init] (initC)   at (0,\concaty)   {\ };
    \node (concat)        at (1.5,\concaty) {$r_1 \cdots r_n$};
    \node[final] (finalC) at (3,\concaty)   {\ };
    \draw[->] (initC)  -- node[above] {$p$} (concat);
    \draw[->] (concat) -- node[above] {$1$} (finalC);

    \node[init] (initCR)   at (4,\concaty) {\ };
    \node (firstCR)        at (5,\concaty) {$r_1$};
    \node (contC)          at (6,\concaty) {$\cdots$};
    \node (lastCR)         at (7,\concaty) {$r_n$};
    \node[final] (finalCR) at (8,\concaty) {\ };
    \draw[->] (initCR)  -- node[above] {$p$} (firstCR);
    \draw[->] (firstCR) -- node[above] {$1$} (contC);
    \draw[->] (contC)   -- node[above] {$1$} (lastCR);
    \draw[->] (lastCR)  -- node[above] {$1$} (finalCR);

    \node[init] (initD)   at (0,3.5) {\ };
    \node (disj)          at (2,3.5) {$r_1 + \dots + r_n$};
    \node[final] (finalD) at (4,3.5) {\ };
    \draw[->] (initD) -- node[above] {$p$} (disj);
    \draw[->] (disj)  -- node[above] {$1$} (finalD);

    \node[init] (initDR)   at (6,3.5) {\ };
    \node (firstDR)        at (7,4.5) {$r_1$};
    \node (cont)           at (7,3.5) {$\vdots$};
    \node (lastDR)         at (7,2.5) {$r_n$};
    \node[final] (finalDR) at (8,3.5) {\ };
    \draw[->] (initDR)  -- node[above] {$p/n\quad$} (firstDR);
    \draw[->] (initDR)  -- node[below] {$p/n\quad$} (lastDR);
    \draw[->] (firstDR) -- node[above] {$\quad 1$} (finalDR);
    \draw[->] (lastDR)  -- node[below] {$\quad 1$} (finalDR);

    \node[init] (initP)   at (0,\plusy) {\ };
    \node (plus)          at (1,\plusy) {$r^{+}$};
    \node[final] (finalP) at (2,\plusy) {\ };
    \draw[->] (initP) -- node[above] {$p$} (plus);
    \draw[->] (plus)  -- node[above] {$1$} (finalP);

    \node[init] (initPR)   at (6,\plusy) {\ };
    \node (plusR)          at (7,\plusy) {$r$};
    \node[final] (finalPR) at (8,\plusy) {\ };
    \draw[->] (initPR) -- node[below] {$p$} (plusR);
    \draw[->] (plusR) .. node[above] {$2/3$}
      controls +(135:0.75) and +(45:0.75) .. (plusR);
    \draw[->] (plusR)  -- node[below] {$1/3$} (finalPR);
    
    \node[init] (initO)   at (0,\optionaly) {\ };
    \node (optional)      at (1,\optionaly) {$r?$};
    \node[final] (finalO) at (2,\optionaly) {\ };
    \draw[->] (initO) -- node[above] {$p$} (optional);
    \draw[->] (optional)  -- node[above] {$1$} (finalO);

    \node[init] (initOR)   at (6,\optionaly) {\ };
    \node (optionalR)      at (7,\optionaly) {$r$};
    \node[final] (finalOR) at (8,\optionaly) {\ };
    \draw[->] (initOR) -- node[below] {$p/2$} (optionalR);
    \draw[->] (initOR) .. node[above] {$p/2$}
      controls +(35:0.5) and +(145:0.5) .. (finalOR);
    \draw[->] (optionalR)  -- node[below] {$1$} (finalOR);
    
  \end{tikzpicture}
  \caption{From a regular expression to a probabilistic automaton.}
  \label{fig:graph-rewriting}
\end{figure}

\myparagraph{Experiments on covering samples} Our first experiment is
designed to test how $\learn$ performs on samples that are at least
large enough to \emph{cover} the target regular expression, in the
following sense.

\begin{definition}
  A sample $S$ \emph{covers} a deterministic automaton $G$ if for
  every edge $(s,t)$ in $G$ there is a word $w \in S$ whose unique
  accepting run in $G$ traverses $(s,t)$. Such a word $w$ is called a
  \emph{witness} for $(s,t)$. A sample $S$ \emph{covers} a
  deterministic regular expression $r$ if it covers the automaton
  obtained from $S$ using the Glushkov construction for translating
  regular expressions into automata as defined in
  Definition~\ref{def:glushkov-translation}.
  \label{def:completeness}
\end{definition}

Intuitively, if a sample does not cover a target regular expression
$r$ then there will be parts of $r$ that cannot be learned from $S$.
In this sense, covering samples are the minimal samples necessary to
learn $r$. Note that such samples are far from ``complete'' or
``characteristic'' in the sense of the theoretical framework of
learning in the limit, as some characteristic samples are bound to be
of size exponential in the size of $r$ by
Theorem~\ref{THM:KORE-EXP-DATA}, while samples of size at most
quadratic in $r$ suffice to cover $r$. Indeed, the Glushkov
construction always yields an automaton whose number of states is
bounded by the size of $r$.  Therefore, this automaton can have at
most $\car{r}^2$ edges, and hence $\car{r}^2$ witness words suffice to
cover $r$.

Table~\ref{tab:suff-data-alpha} shows how \learn performs on covering
samples, broken up by alphabet size of the target expressions.  The
size of the sample used is depicted as well.  The table demonstrates a
remarkable precision. Out of a total of 100 expressions, 82 are
derived exactly for \learn.  Although $\learn(\vldb)$ outperforms
\learn with a success rate of 87 \%, overall $\fixed$ performs best with
89 \%.  The performance decreases with the alphabet size of the target
expressions: this is to be expected since the inference task's
complexity increases.  It should be emphasized that even if $\fixed$
does not derive the target expression exactly, it always yields an
over-approximation, i.e., its language is a superset of the target
language.

Table~\ref{tab:suff-data-size} shows an alternative view on the
results.  It shows the success rate as a function of the target
expression's language size, grouped in intervals.  In particular, it
demonstrates that the method works well for all language sizes.

A final perspective is offered in
Table~\ref{tab:suff-data-states-symbol} which shows the success rate
in function of the average states per symbol $\kappa$ for an
expression.  The latter quantity is defined as the length of the
regular expression excluding operators, divided by the alphabet
size. For instance, for the expression $a (a + b)^{+} c a b$, $\kappa
= 6/3$ since its length excluding operators is 6 and $|\Sigma| = 3$.
It is clear that the learning task is harder for increasing values of
$\kappa$.  To verify the latter, a few extra expressions with large
$\kappa$ values were added to the target expressions.  For the
algorithm $\fixed$ the success rate is quite high for target expressions
with a large value of $\kappa$.  Conversely, $\learn(\vldb)$ yields
better results for $\kappa < 1.6$, while its success rate drops to
around 50 \% for larger values of $\kappa$.  This illustrates that
neither $\learn(\vldb)$ nor $\fixed$ outperforms the other in all
situations.

\begin{table}[tbh]
  \centering
  \begin{tabular}{r|r|r|r|r|r}
    $|\Sigma|$ & \#regex & $\learn(\vldb)$ & \learn & $\fixed$ & $|\Corpus|$ \\ \hline
     5         &  45     &  86 \%     &  97 \%  & 100 \%    &  300  \\
    10         &  45     &  93 \%     &  75 \%  &  84 \%    & 1000  \\
    15         &  10     &  70 \%     &  50 \%  &  60 \%    & 1500  \\ \hline
    total      & 100     &  87 \%     &  82 \%  & \textbf{89 \%} &
  \end{tabular}
  \caption{Success rate on the target regular expressions and
    the sample size used per alphabet size for the various algorithms.}
  \label{tab:suff-data-alpha}
\end{table}

\begin{table}[tbh]
  \centering
  \begin{tabular}{r|r|r|r|r}
    $\Size(r)$   & \#regex &  $\learn(\vldb)$ & \learn &  $\fixed$ \\ \hline
    $[0.0,0.2[$  &    24   & 100 \%          &  87 \% & 96 \% \\
    $[0.2,0.4[$  &    22   &  82 \%          &  91 \% & 91 \% \\
    $[0.4,0.6[$  &    20   &  90 \%          &  75 \% & 85 \% \\
    $[0.6,0.8[$  &    22   &  95 \%          &  72 \% & 83 \% \\
    $[0.8,1.0]$  &    12   &  83 \%          &  78 \% & 78 \%
  \end{tabular}
  \caption{Success rate on the target regular expressions,
    grouped by language size.}
  \label{tab:suff-data-size}
\end{table}

\begin{table}[tbh]
  \centering
  \begin{tabular}{r|r|r|r|r}
    $\kappa$     & \#regex  & $\learn(\vldb)$ & \learn & $\fixed$ \\ \hline
    $[1.2,1.4[$  &    29    &  96 \%          &   72 \% &  83 \% \\
    $[1.4,1.6[$  &    37    & 100 \%          &   89 \% &  89 \% \\
    $[1.6,1.8[$  &    24    &  91 \%          &   92 \% & 100 \% \\
    $[1.8,2.0[$  &    11    &  54 \%          &   91 \% & 100 \% \\
    $[2.0,2.5[$  &    12    &  41 \%          &   50 \% &  50 \% \\
    $[2.5,3.0]$  &    18    &  66 \%          &   71 \% &  78 \%
  \end{tabular}
  \caption{Success rate on the target regular expressions,
    grouped by $\kappa$, the average number of states per symbol.}
  \label{tab:suff-data-states-symbol}
\end{table}

It is also interesting to note that $\learn$ successfully derived the
regular expression $r_1 = (a_1 a_2 + a_3 + \dots + a_n)^+$ of
Theorem~\ref{THM:KORE-EXP-DATA} for $n = 8$, $n=10$, and $n = 12$ from
covering samples of size $500$, $800$, and $1100$, respectively. This
is quite surprising considering that the characteristic samples for
these expressions was proven to be of size at least $(n-2)!$, i.e.,
$720$, $40 320$, and $3 628 800$ respectively. The regular expression
$r_2 = (\alphabet \setminus a_1)^+ a_1 (\alphabet \setminus a_1)^+$,
in contrast, was not derivable by $\learn$ from small samples.

\myparagraph{Experiments on partially covering samples} Unfortunately,
samples to learn regular expressions from are often smaller than one
would prefer.  In an extreme, but not uncommon case, the sample does
not even entirely cover the target expression. In this section we
therefore test how $\learn$ performs on such samples.

\begin{definition}
  The \emph{coverage} of a target regular expression $r$ by a sample
  $S$ is defined as the fraction of transitions in the corresponding
  Glushkov automaton for $r$ that have at least one witness in $S$.
  \label{def:coverage}
\end{definition}

Note that to successfully learn $r$ from a partially covering sample,
$\learn$ needs to ``guess'' the edges for which there is no witness in
$S$.  This guessing capability is built into $\learn(\vldb)$ and
\learn in the form of repair rules~\cite{Bex06,vldbj}. Our experiments
show that for target expressions with alphabet size $|\Sigma| = 10$,
this is highly effective for $\learn(\vldb)$: even at a coverage of
$70 \%$, half the target expressions can still be learned correctly as
Table~\ref{tab:cov10} shows.  The algorithm \learn is performing very
poorly in this setting, being only successful occasionally for
coverages close to 100 \%.  $\fixed$ performs better, although not as
well as $\learn(\vldb)$.  This again illustrates that both
algorithms have their merits.

\begin{table}[tbh]
  \centering
  \begin{tabular}{r|r|r|r}
    coverage & $\learn(\vldb)$ & \learn & $\fixed$    \\ \hline
    1.0      &   100 \%  & 80 \% & 80 \% \\
    0.9      &    64 \%  & 20 \% & 60 \% \\
    0.8      &    60 \%  &  0 \% & 40 \% \\
    0.7      &    52 \%  &  0 \% &  0 \% \\
    0.6      &     0 \%  &  0 \% &  0 \%
  \end{tabular}
  \caption{Success rate for 25 target expressions for $|\Sigma| = 10$
    for samples that provide partial coverage of the target expressions.}
  \label{tab:cov10}
\end{table}

We also experimented with target expressions with alphabet size
$|\Sigma| = 5$.  In this case, the results were not very promising for
$\learn(\vldb)$, but as Table~\ref{tab:cov5} illustrates, \learn and
$\fixed$ performs better, on par with the target expressions for
$|\Sigma| = 10$ in the case of $\fixed$.  This is interesting since the
absolute amount of information missing for smaller regular expressions
is larger than in the case of larger expressions.

\begin{table}[tbh]
  \centering
  \begin{tabular}{r|r|r|r}
    coverage & $\learn(\vldb)$ & \learn & $\fixed$ \\ \hline
    1.0      &   100 \% & 100 \% & 100 \% \\
    0.9      &    25 \% &  75 \% &  66 \% \\
    0.8      &    16 \% &  75 \% &  41 \% \\
    0.7      &     8 \% &  25 \% &  33 \% \\
    0.6      &     8 \% &  25 \% &  17 \% \\
    0.5      &     0 \% &   8 \% &  17 \%
  \end{tabular}
  \caption{Success rate for 12 target expressions for $|\Sigma| = 5$
    with partially covering samples.}
  \label{tab:cov5}
\end{table}



\section{Conclusions}

We presented the algorithm \learn for inferring a deterministic
regular expression from a sample of words.  Motivated by regular
expressions occurring in practice, we use a novel measure based on the
number $k$ of occurrences of the same alphabet symbol and derive
expressions for increasing values of $k$.  We demonstrated the
remarkable effectiveness of \learn on a large corpus of real-world and
synthetic regular expressions of different densities.

Our experiments show that $\learn(\vldb)$ performs better than \learn
for target expressions with a $\kappa < 1.6$ and vice versa for larger
values of $\kappa$.  For partially covering samples, $\learn(\vldb)$
is more robust than \learn.  As $\kappa$ values and sample coverage
are not known in advance, it makes sense to run both algorithms and
select the smallest expression or the one with the smallest language
size, depending on the application at hand.

Some questions need further attention. First, in our experiments,
\learn always derived the correct expression or a super-approximation
of the target expression.  It remains to investigate for which kind of
input samples this behavior can be formally proved.  Second, it would
also be interesting to characterize precisely which classes of
expressions can be learned with our method.  Although the parameter
$\kappa$ explains this to some extend, we probably need more fine
grained measures.  A last and obvious goal for future work is to speed
up the inference of the probabilistic automaton which forms the
bottleneck of the proposed algorithm. A possibility is to use an
industrial strength implementation of the Baum-Welsh algorithm as in
\cite{Finn06} rather than a straightforward one or to explore
different methods for learning probabilistic automata.

Although \learn can be directly plugged into the XSD inference engine
$i$XSD of \cite{Bex07}, it would be interesting to investigate how to
extend these techniques to the more robust class of Relax NG
schemas~\cite{RELAXNG01}.


\bibliographystyle{acmtrans}
\bibliography{cure,refs,tools}

\begin{thebibliography}{}

\bibitem[\protect\citeauthoryear{??}{cas}{}]{castor}
{C}astor.
\newblock \url{www.castor.org}.

\bibitem[\protect\citeauthoryear{??}{jax}{}]{jaxb}
{SUN} {M}icrosystems {JAXB}.
\newblock \url{java.sun.com/webservices/jaxb}.

\bibitem[\protect\citeauthoryear{Adriaans and Vit{\'a}nyi}{Adriaans and
  Vit{\'a}nyi}{2006}]{Adri06a}
{\sc Adriaans, P.} {\sc and} {\sc Vit{\'a}nyi, P.} 2006.
\newblock {T}he {P}ower and {P}erils of {MDL}.

\bibitem[\protect\citeauthoryear{Ahonen}{Ahonen}{1996}]{ahonen}
{\sc Ahonen, H.} 1996.
\newblock {G}enerating {G}rammars for structured documents using grammatical
  inference methods.
\newblock Report A-1996-4, Department of Computer Science, University of
  Finland.

\bibitem[\protect\citeauthoryear{Angluin and Smith}{Angluin and
  Smith}{1983}]{angluinsmith}
{\sc Angluin, D.} {\sc and} {\sc Smith, C.~H.} 1983.
\newblock {I}nductive {I}nference: {T}heory and {M}ethods.
\newblock {\em ACM Computing Surveys\/}~{\em 15,\/}~3, 237--269.

\bibitem[\protect\citeauthoryear{Barbosa, Mignet, and Veltri}{Barbosa
  et~al\mbox{.}}{2005}]{Barb05}
{\sc Barbosa, D.}, {\sc Mignet, L.}, {\sc and} {\sc Veltri, P.} 2005.
\newblock {S}tudying the {XML} {W}eb: gathering statistics from an {XML}
  sample.
\newblock {\em World Wide Web\/}~{\em 8,\/}~4, 413--438.

\bibitem[\protect\citeauthoryear{Benedikt, Fan, and Geerts}{Benedikt
  et~al\mbox{.}}{2005}]{geerts2005}
{\sc Benedikt, M.}, {\sc Fan, W.}, {\sc and} {\sc Geerts, F.} 2005.
\newblock {XP}ath satisfiability in the presence of {DTD}s.
\newblock In {\em {P}roceedings of the {T}wenty-fourth {ACM}
  {SIGACT}-{SIGMOD}-{SIGART} {S}ymposium on {P}rinciples of {D}atabase
  {S}ystems}. 25--36.

\bibitem[\protect\citeauthoryear{Bernstein}{Bernstein}{2003}]{mmbern}
{\sc Bernstein, P.~A.} 2003.
\newblock {A}pplying {M}odel {M}anagement to {C}lassical {M}eta {D}ata
  {P}roblems.
\newblock In {\em First Biennial Conference on Innovative Data Systems
  Research}.

\bibitem[\protect\citeauthoryear{Bex, Neven, Schwentick, and Vansummeren}{Bex
  et~al\mbox{.}}{}]{gjb-sore-journal}
{\sc Bex, G.}, {\sc Neven, F.}, {\sc Schwentick, T.}, {\sc and} {\sc
  Vansummeren, S.}
\newblock {I}nference of {C}oncise {R}egular {E}xpressions and {DTD}s.
\newblock {\em ACM TODS\/}.
\newblock To Appear.

\bibitem[\protect\citeauthoryear{Bex, Gelade, Neven, and Vansummeren}{Bex
  et~al\mbox{.}}{2008}]{Bex08}
{\sc Bex, G.~J.}, {\sc Gelade, W.}, {\sc Neven, F.}, {\sc and} {\sc
  Vansummeren, S.} 2008.
\newblock {L}earning deterministic regular expressions for the inference of
  schemas from {XML} data.
\newblock In {\em WWW}. Beijing, China, 825--834.
\newblock Accepted for WWW 2008.

\bibitem[\protect\citeauthoryear{Bex, Neven, Schwentick, and Tuyls}{Bex
  et~al\mbox{.}}{2006}]{Bex06}
{\sc Bex, G.~J.}, {\sc Neven, F.}, {\sc Schwentick, T.}, {\sc and} {\sc Tuyls,
  K.} 2006.
\newblock {I}nference of concise {DTD}s from {XML} data.
\newblock In {\em Proceedings of the 32nd International Conference on Very
  Large Data Bases}. 115--126.

\bibitem[\protect\citeauthoryear{Bex, Neven, Schwentick, and Vansummeren}{Bex
  et~al\mbox{.}}{2008}]{vldbj}
{\sc Bex, G.~J.}, {\sc Neven, F.}, {\sc Schwentick, T.}, {\sc and} {\sc
  Vansummeren, S.} 2008.
\newblock {I}nference of {C}oncise {R}egular {E}xpressions and {DTD}s.
\newblock submitted to VLDB Journal.

\bibitem[\protect\citeauthoryear{Bex, Neven, and {Van den Bussche}}{Bex
  et~al\mbox{.}}{2004}]{Bex04}
{\sc Bex, G.~J.}, {\sc Neven, F.}, {\sc and} {\sc {Van den Bussche}, J.} 2004.
\newblock {DTD}s versus {XML} {S}chema: a practical study.
\newblock In {\em Proceedings of the 7th International Workshop on the Web and
  Databases}. 79--84.

\bibitem[\protect\citeauthoryear{Bex, Neven, and Vansummeren}{Bex
  et~al\mbox{.}}{2007}]{Bex07}
{\sc Bex, G.~J.}, {\sc Neven, F.}, {\sc and} {\sc Vansummeren, S.} 2007.
\newblock {I}nferring {XML} {S}chema {D}efinitions from {XML} data.
\newblock In {\em {P}roceedings of the 33rd {I}nternational {C}onference on
  {V}ery {L}arge {D}atabases}. 998--1009.

\bibitem[\protect\citeauthoryear{Br{\={a}}zma}{Br{\={a}}zma}{1993}]{Braz93}
{\sc Br{\={a}}zma, A.} 1993.
\newblock {E}fficient identification of regular expressions from representative
  examples.
\newblock In {\em Proceedings of the 6th Annual ACM Conference on Computational
  Learning Theory}. ACM Press, 236--242.

\bibitem[\protect\citeauthoryear{{Br{\"u}ggeman-Klein}}{{Br{\"u}ggeman-Klein}}%
{1993}]{Brug93}
{\sc {Br{\"u}ggeman-Klein}, A.} 1993.
\newblock {R}egular expressions into finite automata.
\newblock {\em Theoretical Computer Science\/}~{\em 120,\/}~2, 197--213.

\bibitem[\protect\citeauthoryear{{Br{\"u}ggemann-Klein} and
  Wood}{{Br{\"u}ggemann-Klein} and Wood}{1998}]{Brue98}
{\sc {Br{\"u}ggemann-Klein}, A.} {\sc and} {\sc Wood, D.} 1998.
\newblock {O}ne-unambiguous regular languages.
\newblock {\em Information and computation\/}~{\em 140,\/}~2, 229--253.

\bibitem[\protect\citeauthoryear{Buneman, Davidson, Fernandez, and
  Suciu}{Buneman et~al\mbox{.}}{1997}]{unstructured}
{\sc Buneman, P.}, {\sc Davidson, S.~B.}, {\sc Fernandez, M.~F.}, {\sc and}
  {\sc Suciu, D.} 1997.
\newblock {A}dding structure to unstructured data.
\newblock In {\em {D}atabase {T}heory - {ICDT} '97, 6th {I}nternational
  {C}onference}, {F.~N. Afrati} {and} {P.~G. Kolaitis}, Eds. Lecture Notes in
  Computer Science, vol. 1186. Springer, 336--350.

\bibitem[\protect\citeauthoryear{Che, Aberer, and {\"O}zsu}{Che
  et~al\mbox{.}}{2006}]{tamer}
{\sc Che, D.}, {\sc Aberer, K.}, {\sc and} {\sc {\"O}zsu, M.~T.} 2006.
\newblock {Q}uery optimization in {XML} structured-document databases.
\newblock {\em VLDB Journal\/}~{\em 15,\/}~3, 263--289.

\bibitem[\protect\citeauthoryear{Chidlovskii}{Chidlovskii}{2001}]{DBLP:conf/kr%
db/Chidlovskii01}
{\sc Chidlovskii, B.} 2001.
\newblock {S}chema extraction from {XML}: a grammatical inference approach.
\newblock In {\em {P}roceedings of the 8th {I}nternational {W}orkshop on
  {K}nowledge {R}epresentation meets {D}atabases}.

\bibitem[\protect\citeauthoryear{Clark}{Clark}{}]{trang}
{\sc Clark, J.}
\newblock {T}rang: {M}ulti-format schema converter based on {RELAX NG}.
\newblock \url{http://www.thaiopensource.com/relaxng/trang.html}.

\bibitem[\protect\citeauthoryear{Clark and Murata}{Clark and
  Murata}{2001}]{RELAXNG01}
{\sc Clark, J.} {\sc and} {\sc Murata, M.} 2001.
\newblock {\em {RELAX} {NG} {S}pecification}.
\newblock OASIS.

\bibitem[\protect\citeauthoryear{Cover}{Cover}{2003}]{Cover03}
{\sc Cover, R.} 2003.
\newblock {T}he {C}over {P}ages.
\newblock http://xml.coverpages.org/.

\bibitem[\protect\citeauthoryear{Du, {Amer-Yahia}, and Freire}{Du
  et~al\mbox{.}}{2004}]{shrex}
{\sc Du, F.}, {\sc {Amer-Yahia}, S.}, {\sc and} {\sc Freire, J.} 2004.
\newblock {S}hre{X}: {M}anaging {XML} {D}ocuments in {R}elational {D}atabases.
\newblock In {\em {P}roceedings of the 30th {I}nternational {C}onference on
  {V}ery {L}arge {D}ata {B}ases}. 1297--1300.

\bibitem[\protect\citeauthoryear{Ehrenfeucht and Zeiger}{Ehrenfeucht and
  Zeiger}{1976}]{Ehre76}
{\sc Ehrenfeucht, A.} {\sc and} {\sc Zeiger, P.} 1976.
\newblock {C}omplexity measures for regular expressions.
\newblock {\em Journal of computer and system sciences\/}~{\em 12}, 134--146.

\bibitem[\protect\citeauthoryear{Fernau}{Fernau}{2004}]{Fern04}
{\sc Fernau, H.} 2004.
\newblock {E}xtracting minimum length {D}ocument {T}ype {D}efinitions is
  {NP}-hard.
\newblock In {\em ICGI}. 277--278.

\bibitem[\protect\citeauthoryear{Fernau}{Fernau}{2005}]{fernaualt}
{\sc Fernau, H.} 2005.
\newblock {A}lgorithms for {L}earning {R}egular {E}xpressions.
\newblock In {\em {A}lgorithmic {L}earning {T}heory, 16th {I}nternational
  {C}onference}. 297--311.

\bibitem[\protect\citeauthoryear{Finn, Mistry, Schuster-Böckler,
  Griffiths-Jones, et~al\mbox{.}}{Finn et~al\mbox{.}}{2006}]{Finn06}
{\sc Finn, R.}, {\sc Mistry, J.}, {\sc Schuster-Böckler, B.}, {\sc
  Griffiths-Jones, S.}, {\sc et~al\mbox{.}} 2006.
\newblock {P}fam: clans, web tools and services.
\newblock {\em Nucleic Acids Research\/}~{\em 34}, D247--D251.

\bibitem[\protect\citeauthoryear{Florescu}{Florescu}{2005}]{Flor05}
{\sc Florescu, D.} 2005.
\newblock {M}anaging semi-structured data.
\newblock {\em ACM Queue\/}~{\em 3,\/}~8 (October).

\bibitem[\protect\citeauthoryear{Fran\c{c}ois}{Fran\c{c}ois}{2006}]{Jahmm}
{\sc Fran\c{c}ois, J.-M.} 2006.
\newblock {J}ahmm.
\newblock \url{http://www.run.montefiore.ulg.ac.be/~francois/software/jahmm/}.

\bibitem[\protect\citeauthoryear{Freire, Haritsa, Ramanath, Roy, and
  Sim{\'e}on}{Freire et~al\mbox{.}}{2002}]{Frei02}
{\sc Freire, J.}, {\sc Haritsa, J.~R.}, {\sc Ramanath, M.}, {\sc Roy, P.}, {\sc
  and} {\sc Sim{\'e}on, J.} 2002.
\newblock {S}tati{X}: making {XML} count.
\newblock In {\em SIGMOD Conference}. 181--191.

\bibitem[\protect\citeauthoryear{Freitag and McCallum}{Freitag and
  McCallum}{2000}]{Frei00}
{\sc Freitag, D.} {\sc and} {\sc McCallum, A.} 2000.
\newblock {I}nformation {E}xtraction with {HMM} {S}tructures {L}earned by
  {S}tochastic {O}ptimization.
\newblock In {\em AAAI/IAAI}. AAAI Press / The MIT Press, 584--589.

\bibitem[\protect\citeauthoryear{Garcia and Vidal}{Garcia and
  Vidal}{1990}]{Garc90}
{\sc Garcia, P.} {\sc and} {\sc Vidal, E.} 1990.
\newblock {I}nference of k-testable languages in the strict sense and
  application to syntactic pattern recognition.
\newblock {\em IEEE Transactions on Pattern Analysis and Machine
  Intelligence\/}~{\em 12,\/}~9 (September), 920--925.

\bibitem[\protect\citeauthoryear{Garofalakis, Gionis, Rastogi, Seshadri, and
  Shim}{Garofalakis et~al\mbox{.}}{2003}]{Garo03}
{\sc Garofalakis, M.}, {\sc Gionis, A.}, {\sc Rastogi, R.}, {\sc Seshadri, S.},
  {\sc and} {\sc Shim, K.} 2003.
\newblock {XTRACT}: learning document type descriptors from {XML} document
  collections.
\newblock {\em Data mining and knowledge discovery\/}~{\em 7}, 23--56.

\bibitem[\protect\citeauthoryear{Gelade and Neven}{Gelade and
  Neven}{2008}]{stacs08}
{\sc Gelade, W.} {\sc and} {\sc Neven, F.} 2008.
\newblock {S}uccinctness of the {C}omplement and {I}ntersection of {R}egular
  {E}xpressions.
\newblock In {\em STACS}. 325--336.

\bibitem[\protect\citeauthoryear{Gold}{Gold}{1967}]{Gold67}
{\sc Gold, E.} 1967.
\newblock {L}anguage identification in the limit.
\newblock {\em Information and Control\/}~{\em 10,\/}~5 (May), 447--474.

\bibitem[\protect\citeauthoryear{Goldman and Widom}{Goldman and
  Widom}{1997}]{dataguides}
{\sc Goldman, R.} {\sc and} {\sc Widom, J.} 1997.
\newblock {D}ata{G}uides: {E}nabling {Q}uery {F}ormulation and {O}ptimization
  in {S}emistructured {D}atabases.
\newblock In {\em Proceedings of 23rd International Conference on Very Large
  Data Bases}. 436--445.

\bibitem[\protect\citeauthoryear{Gruber and Holzer}{Gruber and
  Holzer}{2008}]{DBLP:conf/icalp/GruberH08}
{\sc Gruber, H.} {\sc and} {\sc Holzer, M.} 2008.
\newblock {F}inite {A}utomata, {D}igraph {C}onnectivity, and {R}egular
  {E}xpression {S}ize.
\newblock In {\em ICALP (2)}. 39--50.

\bibitem[\protect\citeauthoryear{Hegewald, Naumann, and Weis}{Hegewald
  et~al\mbox{.}}{2006}]{Hege06}
{\sc Hegewald, J.}, {\sc Naumann, F.}, {\sc and} {\sc Weis, M.} 2006.
\newblock {XS}truct: efficient schema extraction from multiple and large {XML}
  documents.
\newblock In {\em ICDE Workshops}. 81.

\bibitem[\protect\citeauthoryear{Hopcroft and Ullman}{Hopcroft and
  Ullman}{2007}]{Hopc79}
{\sc Hopcroft, J.} {\sc and} {\sc Ullman, J.} 2007.
\newblock {\em {I}ntroduction to automata theory, languages and computation}.
\newblock Addison-Wesley, Reading, MA.

\bibitem[\protect\citeauthoryear{Koch, Scherzinger, Schweikardt, and
  Stegmaier}{Koch et~al\mbox{.}}{2004}]{Koch04b}
{\sc Koch, C.}, {\sc Scherzinger, S.}, {\sc Schweikardt, N.}, {\sc and} {\sc
  Stegmaier, B.} 2004.
\newblock {S}chema-based scheduling of event processors and buffer minimization
  for queries on structured data streams.
\newblock In {\em {P}roceedings of the 30th {I}nternational {C}onference on
  {V}ery {L}arge {D}ata {B}ases}. 228--239.

\bibitem[\protect\citeauthoryear{Manolescu, Florescu, and Kossmann}{Manolescu
  et~al\mbox{.}}{2001}]{kossmann}
{\sc Manolescu, I.}, {\sc Florescu, D.}, {\sc and} {\sc Kossmann, D.} 2001.
\newblock {A}nswering {XML} {Q}ueries on {H}eterogeneous {D}ata {S}ources.
\newblock In {\em {P}roceedings of 27th {I}nternational {C}onference on {V}ery
  {L}arge {D}ata {B}ases}. 241--250.

\bibitem[\protect\citeauthoryear{Martens, Neven, Schwentick, and Bex}{Martens
  et~al\mbox{.}}{2006}]{Mart06b}
{\sc Martens, W.}, {\sc Neven, F.}, {\sc Schwentick, T.}, {\sc and} {\sc Bex,
  G.~J.} 2006.
\newblock {E}xpressiveness and {C}omplexity of {XML} {S}chema.
\newblock {\em ACM Transactions on Database Systems\/}~{\em 31,\/}~3, 770--813.

\bibitem[\protect\citeauthoryear{Mignet, Barbosa, and Veltri}{Mignet
  et~al\mbox{.}}{2003}]{Mign03}
{\sc Mignet, L.}, {\sc Barbosa, D.}, {\sc and} {\sc Veltri, P.} 2003.
\newblock {T}he {XML} web: a first study.
\newblock In {\em Proceedings of the 12th International World Wide Web
  Conference}. Budapest, Hungary, 500--510.

\bibitem[\protect\citeauthoryear{Nestorov, Abiteboul, and Motwani}{Nestorov
  et~al\mbox{.}}{1998}]{nestorov}
{\sc Nestorov, S.}, {\sc Abiteboul, S.}, {\sc and} {\sc Motwani, R.} 1998.
\newblock {E}xtracting {S}chema from {S}emistructured {D}ata.
\newblock In {\em International Conference on Management of Data}. ACM Press,
  295--306.

\bibitem[\protect\citeauthoryear{Neven and Schwentick}{Neven and
  Schwentick}{2006}]{nevenschwentickicdt03}
{\sc Neven, F.} {\sc and} {\sc Schwentick, T.} 2006.
\newblock {O}n the complexity of {XP}ath containment in the presence of
  disjunction, {DTD}s, and variables.
\newblock {\em Logical Methods in Computer Science\/}~{\em 2,\/}~3.

\bibitem[\protect\citeauthoryear{Pitt}{Pitt}{1989}]{Pitt89}
{\sc Pitt, L.} 1989.
\newblock {I}nductive {I}nference, {DFA}s, and {C}omputational {C}omplexity.
\newblock In {\em Proceedings of the International Workshop on Analogical and
  Inductive Inference}, {K.~P. Jantke}, Ed. Lecture Notes in Computer Science,
  vol. 397. Springer-Verlag, 18--44.

\bibitem[\protect\citeauthoryear{Quass, Widom, Goldman, et~al\mbox{.}}{Quass
  et~al\mbox{.}}{1996}]{lore}
{\sc Quass, D.}, {\sc Widom, J.}, {\sc Goldman, R.}, {\sc et~al\mbox{.}} 1996.
\newblock {LORE}: a {L}ightweight {O}bject {RE}pository for semistructured
  data.
\newblock In {\em Proceedings of the 1996 ACM SIGMOD International Conference
  on Management of Data}. 549.

\bibitem[\protect\citeauthoryear{Rabiner}{Rabiner}{1989}]{Rabi89}
{\sc Rabiner, L.} 1989.
\newblock {A} tutorial on {H}idden {M}arkov {M}odels and selected applications
  in speech recognition.
\newblock {\em Proc. IEEE\/}~{\em 77,\/}~2, 257--286.

\bibitem[\protect\citeauthoryear{Rahm and Bernstein}{Rahm and
  Bernstein}{2001}]{schemamatching}
{\sc Rahm, E.} {\sc and} {\sc Bernstein, P.~A.} 2001.
\newblock {A} survey of approaches to automatic schema matching.
\newblock {\em VLDB Journal\/}~{\em 10,\/}~4, 334--350.

\bibitem[\protect\citeauthoryear{Sahuguet}{Sahuguet}{2000}]{Sahu00}
{\sc Sahuguet, A.} 2000.
\newblock {E}verything {Y}ou {E}ver {W}anted to {K}now {A}bout {DTD}s, {B}ut
  {W}ere {A}fraid to {A}sk ({E}xtended {A}bstract).
\newblock In {\em The World Wide Web and Databases, 3rd International
  Workshop}, {D.~Suciu} {and} {G.~Vossen}, Eds. Lecture Notes in Computer
  Science, vol. 1997. Springer, 171--183.

\bibitem[\protect\citeauthoryear{Sakakibara}{Sakakibara}{1997}]{Saka97}
{\sc Sakakibara, Y.} 1997.
\newblock {R}ecent advances of grammatical inference.
\newblock {\em Theoretical Computer Science\/}~{\em 185,\/}~1, 15--45.

\bibitem[\protect\citeauthoryear{Sankey and Wong}{Sankey and
  Wong}{2001}]{wongDTD}
{\sc Sankey, J.} {\sc and} {\sc Wong, R.~K.} 2001.
\newblock {S}tructural inference for semistructured data.
\newblock In {\em Proceedings of the 10th international conference on
  Information and knowledge management}. ACM Press, 159--166.

\bibitem[\protect\citeauthoryear{Thompson, Beech, Maloney, and
  Mendelsohn}{Thompson et~al\mbox{.}}{2001}]{XSDS01}
{\sc Thompson, H.}, {\sc Beech, D.}, {\sc Maloney, M.}, {\sc and} {\sc
  Mendelsohn, N.} 2001.
\newblock {\em {XML} {S}chema part 1: structures}.
\newblock W3C.

\bibitem[\protect\citeauthoryear{Young-Lai and Tompa}{Young-Lai and
  Tompa}{2000}]{Youn00}
{\sc Young-Lai, M.} {\sc and} {\sc Tompa, F.~W.} 2000.
\newblock {S}tochastic {G}rammatical {I}nference of {T}ext {D}atabase
  {S}tructure.
\newblock {\em Machine Learning\/}~{\em 40,\/}~2, 111--137.

\end{thebibliography}


\begin{received} 
Received Month Year; revised Month Year; accepted Month Year 
\end{received} 

\end{document}